\documentclass[twoside]{article}
\usepackage[a4paper]{geometry}
\usepackage[utf8]{inputenc} 
\usepackage[T1]{fontenc} 
\usepackage{RR}
\usepackage{hyperref}
\RRNo{8491}
\RRdate{March 2014}
\RRauthor{Andr\'{e} Tavares\thanks{UFMG}
\and Benoit Boissinot\thanks{ENS Lyon}
\and Fernando Pereira\thanks{UFMG} 
\and Fabrice Rastello\thanks{Inria}
}
\authorhead{Tavares, Boissinot, Pereira, and Rastello}
\RRetitle{Parameterized Construction of Program Representations for Sparse
Dataflow Analyses}

\RRtitle{Représentation de programmes pour l'analyse creuse de flots de données: construction paramétrée}

\RRabstract{Data-flow analyses usually associate information with control flow regions.
Informally, if these regions are too small, like a point between two
consecutive statements, we call the analysis dense.
On the other hand, if these regions include many such points, then we call it
sparse.
This paper presents a systematic method to build program representations that
support sparse analyses.
To pave the way to this framework we clarify the
bibliography about well-known intermediate program representations.
We show that our approach, up to parameter choice, subsumes many of these
representations, such as the SSA, SSI and e-SSA forms.
In particular, our algorithms are faster, simpler and more frugal than the
previous techniques used to construct SSI - Static Single Information - form
programs.
We produce intermediate representations isomorphic to  Choi
{\em et al.}'s Sparse Evaluation Graphs (SEG) for the family of data-flow
problems that can be partitioned per variables.
However, contrary to SEGs, we can handle - sparsely - problems that are not in
this family.
}
\RRresume{
L'analyse de flot de données, associe en général l'information calculée, aux régions de flot de contrôle.
Informellement cette analyse est dite dense, si ces régions sont trop petites, i.e. par exemple restreintes aux points de programme situés entre deux instructions.
A l'opposé, cette analyse est dite creuse, si ces régions comprennent de nombreux points consécutifs.
Cet article présente une méthode de construction systématique d'une représentation de programme qui permet de manière naturelle l'implémentation d'analyses creuses.
Cette forme englobe plusieurs forme existante comme la forme SSA, la forme SSI, ou la forme e-SSA.
En particulier, l'algorithme présenté est plus rapide, plus simple et moins gourmand que les méthodes existantes de construction de SSI --Static Single Information.
Aussi, la représentation ainsi construite se trouve être isomorphe au graphe d'évaluation creux (Sparse Evaluation Graph --- SEG in English) de Choi et al. dans le cas particulier ou le problème d'analyse de flot de données peut être partitionné par variable. 
Cela dit, contrairement aux SEG, l'approche ici décrite n'est pas restreinte à cette famille de problèmes.
 }

\RRkeyword{Sparse Data-Flow Analysis, Compiler, Static Single Assignment, Static Single Information, SSA, SSI, Static Single Use, SSU, Iterated Dominance Frontier, Control-Flow Graph}
\RRmotcle{Analysis de flot de données, compilateur, forme à assignation unique, SSA, SSI, SSU, frontière de dominance itérée, graphe de flot de contrôle}

\RRprojets{GCG, Compsys}
\URRhoneAlpes

\usepackage{amsmath}
\usepackage{amsthm}
\usepackage{amssymb}
\usepackage{graphicx}
\usepackage{subfig}
\usepackage{multirow}
\usepackage{tabularx}
\usepackage[table]{xcolor}
\usepackage{rotating}
\usepackage{hhline}
\usepackage{xspace}

\newtheorem{definition}{Definition}
\newtheorem{theorem}{Theorem}
\newtheorem{lemma}{Lemma}

\def\SS{{\cal P}}
\def\Out{\mathrm{Out}}
\def\In{\mathrm{In}}
\def\Defs{\mathrm{Defs}}
\def\Def{\mathrm{Def}}
\def\Uses{\mathrm{Uses}}

\def\progpoint{program point\xspace}
\def\progpoints{program points\xspace}
\def\splitpoint{control flow node\xspace}
\def\splitpoints{control flow nodes\xspace}
\def\inclu{\sqsubseteq}
\def\undef{\textsf{undef}}
\def\SSIfy{\textsf{SSIfy}}
\def\incoming{\textrm{incoming\_edges}}
\def\outgoing{\textrm{outgoing\_edges}}
\def\isjoin{\textrm{is\_join}}
\def\isfork{\textrm{is\_fork}}
\def\Out{\mathrm{Out}}
\def\1{\qquad}
\def\2{\1\1}
\def\3{\2\1}
\def\4{\2\2}
\def\5{\3\2}
\def\6{\4\2}
\def\7{\5\2}
\def\8{\6\2}
\def\9{\7\2}
\def\If{{\sf  if }}
\def\Let{{\sf  let }}
\def\Then{{\sf  then }}
\def\Else{{\sf  else}}
\def\Foreach{{\sf foreach }}
\def\For{{\sf for }}
\def\While{{\sf while }}
\newcommand\var[1]{\mbox{\em #1}}
\begin{document}

\makeRR

\section{Introduction}
\label{sec:int}

Many data-flow analyses bind information to pairs formed by a variable and a
\progpoint~\cite{Ackerman84,Bodik00,Cartwright89,Damas82,Johnson93,Mahlke01,Nanda09,Plevyak96,Rimsa11,Roy10,Scholz08,Stephenson00,Su05,Hochstadt08,Wegman91}.
As an example, for each \progpoint $p$, and each integer variable $v$ live at
$p$, Stephenson {\em et al.}'s~\cite{Stephenson00} bit-width analysis finds the
size, in bits, of $v$ at $p$.
Although well studied in the literature, this approach might
produce redundant information.
For instance, a given variable $v$ may be mapped to the same bit-width along many
consecutive \progpoints.
Therefore, a natural way to reduce redundancies is to make these analyses
{\em sparser}, increasing the granularity of the program regions that they
manipulate.

There exists different attempts to implement data-flow analyses sparsely.
The Static Single Assignment (SSA) form~\cite{Cytron91}, for instance, allows us
to implement several analyses and optimizations, such as reaching definitions
and constant propagation, sparsely.
Since its conception, the SSA format has been generalized into many different
program representations, such as the {\em Extended-SSA} form~\cite{Bodik00}, the
{\em Static Single Information} (SSI) form~\cite{Ananian99},
and the {\em Static Single Use} (SSU)
form~\cite{George03,Lo98,Plevyak96}.
Each of these representations extends the reach of the SSA form to sparser
data-flow analyses; however, there is not a format that subsumes all the
others.
In other words, each of these three program representations fit specific types
of data-flow problems.
Another attempt to model data-flow analyses sparsely is due to Choi
{\em et al.}'s {\em Sparse Evaluation Graph} (SEG)~\cite{Choi91}.
This data-structure supports several different analyses sparsely, as long as
the abstract state of a variable does not interfere with the abstract state of
other variables in the same program.
This family of analyses is known as {\em Partitioned Variable Problems} in
the literature~\cite{Zadeck84}.

In this paper, we propose a framework that includes all
these previous approaches.
Given a data-flow problem defined by (i) a set of \splitpoints, that produce information, and (ii) a direction in which
information flows: forward, backward or both ways, we build a program
representation that allows to solve the problem sparsely using def-use chains.
The program representations that we generate ensure a key {\em single
information property}: the data-flow facts associated with a variable are
invariant along the entire live range of this variable.

\section{Static Single Information}
\label{sec_ssi}

Our objective is to generate program representations that bestow the {\em Static
Single Information property} (Definition~\ref{def:ssi}) onto a given data-flow
problem.
In order to introduce this notion, we will need a number of concepts, which we
define in this chapter.
We start with the concept of a {\em Data-Flow System}, which
Definition~\ref{def:dfa} recalls from the literature.
We consider a {\em \progpoint} a point between two consecutive instructions. If $p$ is a \progpoint, then $\mathit{preds}(p)$
(resp. $\mathit{succs}(p)$) is the set of all the \progpoints that are predecessors (resp. successors) of $p$.
A {\em transfer function} determines how information flows among these
\progpoints.
Information are elements of a {\em lattice}.
We find a solution to a data-flow problem by continuously solving the set of
transfer functions associated with each program region until a fix point
is reached.
Some \progpoints are {\em meet nodes}, because they combine
information coming from two or more regions.
The result of combining different elements of a lattice is given by a {\em meet}
operator, which we denote by $\wedge$.

\begin{definition}[Data-Flow System]
\label{def:dfa}
A data-flow system $E_{\var{dense}}$ is an equation system that associates, with each
\progpoint $p$, an element of a lattice ${\cal L}$, given by the equation
$x^p = \bigwedge_{s \in \mathit{preds}(p)} F^{s,p}(x^s)$, where: $x^p$ denotes
the abstract state associated with \progpoint $p$; $\mathit{preds}(p)$ is the set of control flow predecessors of $p$; $F^{s,p}$ is
the transfer function from \progpoint $s$ to \progpoint $p$.
The analysis can alternatively be written as a constraint system that binds to each \progpoint $p$ and each $s\in \mathit{preds}(p)$ the equation
$x^p = x^p \wedge  F^{s,p}(x^s)$ or, equivalently, the inequation
$x^p \inclu  F^{s,p}(x^s)$.
\end{definition}

The program representations that we generate lets us solve a class of data-flow
problems that we call {\em Partitioned Lattice per Variable} (PLV), and that
we introduce in Definition~\ref{def:plv}.
Constant propagation is an example of a PLV problem.
If we denote by $\cal C$ the lattice of constants, the overall lattice can be
written as ${\cal L}={\cal C}^n$, where $n$ is the number of variables.
In other words, this data-flow problem ranges on a product lattice that contains
a term for each variable in the target program.

\begin{definition}[Partitioned Lattice per Variable Problem (PLV)]
\label{def:plv}
Let ${\cal V}=\{v_1,\dots,v_n\}$ be the set of program variables. 
The Maximum Fixed Point problem on a data-flow system is a \emph{Partitioned
Lattice per Variable Problem} if, and only if, $\cal L$ can be decomposed
into the product of ${\cal L}_{v_1}\times \dots \times {\cal L}_{v_n}$ where
each ${\cal L}_{v_i}$ is the lattice  associated with program variable $v_i$.
In other words $x^s$ can be writen as $([v_1]^s,\dots,[v_n]^s)$ where $[v]^s$ denotes the abstract state associated with variable $v$ and \progpoint $s$. $F^{s,p}$ can thus be decomposed into the product of $F^{s,p}_{v_1}\times \dots\times F^{s,p}_{v_n}$ and the constraint system decomposed into the inequalities $[v_i]^p\sqsubseteq  F^{s,p}_{v_i}([v_1]^s,\dots,[v_n]^s)$.
\end{definition}

The transfer functions that we describe in Definition~\ref{def:tf} have no
influence on the solution of a data-flow system.
The goal of a sparse data-flow analysis is to shortcut these functions.
We accomplish this task by grouping contiguous \progpoints bound to these
functions into larger regions.

\begin{definition}[Trivial/Constant/Undefined Transfer functions]
\label{def:tf}
Let ${\cal L}_{v_1}\times {\cal L}_{v_2} \times \dots \times {\cal L}_{v_n}$ be
the decomposition per variable of lattice $\cal L$, where ${\cal L}_{v_i}$ is
the lattice associated with variable $v_i$.
Let $F_{v_i}$ be a transfer function from ${\cal L}$ to ${\cal L}_{v_i}$.
\begin{itemize}
\item $F_{v_i}$ is \emph{trivial} if $\forall x=([v_1], \dots, [v_n])\in {\cal L},\, F_{v_i}(x)=[v_i]$
\item $F_{v_i}$ is \emph{constant with value $C\in {\cal L}_{v_i}$} if $\forall x\in {\cal L},\, F_{v_i}(x)=C$
\item $F_{v_i}$ is {\em undefined} if $F_{v_i}$ is constant with value $\top$,
e.g., $F_{v_i}(x) = \top$, where $\top \wedge y = y \wedge \top = y$.
\end{itemize}
\end{definition}

A sparse data-flow analysis propagates information from the \splitpoint where this
information is created directly to the \splitpoint where this information is needed.
Therefore, the notion of {\em dependence}, which we state in
Definition~\ref{def:dep}, plays a fundamental role in our framework.
Intuitively, we say that a variable $v$ depends on a variable $v_j$ if the
information associated with $v$ might change in case the information associated
with $v_j$ does.

\begin{definition}[Dependence]
\label{def:dep}
We say that $F_v$ \emph{depends on variable $v_j$} if:
$$\begin{array}{l}
\exists x=([v_1],\dots,[v_n])\neq ([v_1]',\dots,[v_n]')=x' \textrm{ in } {\cal L}\\ 
\textrm{such that } \left[\strut F_v(x)\neq F_v(x') \textrm{ and } \forall k\neq j, \ [v_k] = [v_k]'\right]
\end{array}
$$
\end{definition}

In a {\em backward} data-flow analysis, the information that comes from
the predecessors of a node $n$ is combined to produce the information that
reaches the successors of $n$.
A {\em forward} analysis propagates information in the opposite direction.
We call meet nodes those places where information coming from multiple sources
are combined.
Definition~\ref{def:merge} states this concept more formally.

\begin{definition}[Meet Nodes]
\label{def:merge}
Consider a forward (resp. backward) monotone PLV problem, where $(Y_v^p)$ is
the maximum fixed point solution of variable $v$ at \progpoint $p$.
We say that a \progpoint $p$ is a meet node for variable $v$ if, and only if,
$p$ has $n\geq 2$ predecessors (resp. successors), $s_1, \ldots, s_n$, and there
exists $s_i\neq s_j$, such that $Y_v^{s_i} \neq Y_v^{s_j}$.
\end{definition}

Our goal is to build program representations in which the information
associated with a variable is invariant along the entire live range of this
variable.
A variable $v$ is {\em alive} at a \progpoint $p$ if there is a path from
$p$ to an instruction that uses $v$, and $v$ is not re-defined along the way.
The live range of $v$, which we denote by {\em live(v)}, is the
collection of \progpoints where $v$ is alive.

\begin{definition}[Static Single Information property]
\label{def:ssi}
Consider a forward (resp. backward) monotone PLV problem $E_{\var{dense}}$ stated as
in Definition~\ref{def:dfa}.
A program representation fulfills the Static Single Information property if, and
only if, it meets the following properties for each variable $v$:

\begin{description}
\item {\bf [SPLIT-DEF]:} for each two consecutive \progpoints $s$ and $p$ (resp. $p$ and $s$) such that $p\in \textrm{live}(v)$, and $F_v^{s,p}$ is non-trivial nor undefined, there should be an instruction between $s$ and $p$ that contains a definition (resp. last use) of $v$;
\item {\bf [SPLIT-MEET]:} each meet node $p$ with $n$ predecessors
$\{s_1, \ldots, s_n\}$ (resp. successors) should have a definition (resp. use)
of $v$ at $p$, and $n$ uses (resp. definitions) of $v$, one at each $s_i$.
We shall implement these defs/uses with $\phi$/$\sigma$-functions, as we
explain in Section~\ref{sub:split}.
\item {\bf [INFO]:} each \progpoint $p\not\in \textrm{live}(v)$ should be bound to undefined transfer functions, e.g., $F_v^{s,p}=\lambda x.\top$ for each $s\in \textit{preds}(p)$ (resp. $s\in \textit{succs}(p)$).
\item {\bf [LINK]:} for each two consecutive \progpoints $s$ and $p$  (resp. $p$ and $s$)  for which $F_v^\textit{s,p}$ depends on some $[u]^s$, there should be an instruction between $s$ and $p$ that contains a (potentially pseudo) use (resp. def) of $u$.
\item {\bf [VERSION]:} for each variable $v$, $\textrm{live}(v)$ is a connected component of the CFG.
\end{description}
\end{definition}

\subsection{Special instructions used to split live ranges}
\label{sub:split}

We group \splitpoints in three kinds:
interior nodes, forks and joins.
At each place we use a different notation to denote live range splitting.

{\em Interior nodes} are \splitpoints that have a unique predecessor and a unique successor. 
At these \splitpoints we perform live range splitting via copies.
If the \splitpoint already contains another instruction, then this copy \emph{must} be done \emph{in parallel} with the existing instruction.
The notation, \[\var{inst} \ \parallel\  v_1=v'_1 \ \parallel\  \dots \ \parallel\  v_m=v'_m\] denotes $m$ copies $v_i=v'_i$ performed in parallel with
instruction \var{inst}.
This means that all the uses of \var{inst} plus all $v'_i$ are read simultaneously, then \var{inst} is computed, then all definitions of \var{inst} plus all $v_i$ are written simultaneously.

In forward analyses, the information produced at different definitions of a
variable may reach the same meet node.
To avoid that these definitions reach the same use of $v$, we merge them at the earliest \splitpoint where they meet; hence, ensuring [SPLIT-MEET].
We do this merging via special instructions called $\phi$-functions, which were introduced by Cytron {\em et al.} to build SSA-form programs~\cite{Cytron91}.
The assignment \[v_1 = \phi(l^1:v_1^1, \ldots, l^q:v_1^q) \ \parallel\  \dots \ \parallel\  v_m = \phi(l^1:v_m^1, \ldots, l^q:v_m^q)\] contains $m$ $\phi$-functions to be performed in parallel.
The $\phi$ symbol works as a multiplexer.
It will assign to each $v_i$ the value in $v_i^j$, where $j$ is determined by $l^j$, the basic block last visited before reaching the $\phi$-function.
The above statement encapsulates $m$ parallel copies: all the variables
$v_1^j, \ldots, v_m^j$ are simultaneously copied into the variables
$v_1, \ldots, v_m$.
{Note that our notion of \splitpoints differs from the usual notion of nodes of the CFG}. A join node actually corresponds to the entry point of a CFG node: to this end we denote as $\In(l)$ the point right before $l$. As an example in Figure~\ref{fig:classInference}(d), $l_7$ is considered to be an interior node, and the $\phi$-function defining $v_6$ has been inserted at the join node $\In(l_7)$.

In backward analyses the information that emerges from different uses of a variable may reach the same meet node.
To ensure Property [SPLIT-MEET], the use that reaches the definition of a
variable must be unique, in the same way that in a SSA-form program the definition that reaches a use is unique.
We ensure this property via special instructions that Ananian has called
$\sigma$-functions~\cite{Ananian99}.
The $\sigma$-functions are the symmetric of $\phi$-functions, performing a parallel assignment depending on the execution path taken.
The assignment \[(l^1:v_1^1, \ldots, l^q:v_1^q) = \sigma(v_1) \ \parallel\  \dots \ \parallel\  (l^1:v_m^1, \ldots, l^q:v_m^q) = \sigma(v_m)\] represents $m$ $\sigma$-functions that assign to each variable $v_i^j$ the value in $v_i$ if control flows into block $l^j$.
These assignments happen in parallel, i.e., the $m$ $\sigma$-functions encapsulate $m$ parallel copies.
Also, notice that variables live in different branch targets are given
different names by the $\sigma$-function that ends that basic block.
Similarly to join nodes, a fork node is the exit point of a CFG node: $\Out(l)$ denotes the point right after CFG node $l$. As an example in Figure~\ref{fig:classInference}(d), $l_2$ is considered to be an interior node, and the $\sigma$-function using $v_1$ has been inserted at the fork node $\Out(l_2)$.  

\subsection{Examples of PLV Problems}
\label{sub:examples}

Many data-flow analyses can be classified as PLV problems.
In this section we present some meaningful examples.
Along each example we show the program representation that lets us solve it
sparsely.

~\\
\noindent
\textbf{Class Inference:} Some dynamically typed languages, such as Python,
Java\-Scrip, Ruby or Lua, represent objects as hash tables containing methods
and fields.
In this world, it is possible to speedup execution by replacing these
hash tables with actual object oriented virtual tables.
A class inference engine tries to assign a virtual table to a variable $v$
based on the ways that $v$ is used.
The Python program in Figure~\ref{fig:classInference}(a) illustrates this
optimization.
Our objective is to infer the correct suite of methods for each object bound to
variable $v$.
Figure~\ref{fig:classInference}(b) shows the control flow graph of the program,
and Figure~\ref{fig:classInference}(c) shows the results of a dense
implementation of this analysis.
In a dense analysis, each program instruction is associated with a transfer
function; however, some of these functions, such as that in label $l_3$, are
trivial.
We produce, for this example, the representation
given in Figure~\ref{fig:classInference}(d).
Because type inference is a backward analysis that extracts information from
use sites, we split live ranges at these \splitpoints, and rely on
$\sigma$-functions to merge them back.
The use-def chains that we derive from the program representation,
seen in Figure~\ref{fig:classInference}(e), lead naturally
to a constraint system, which we show in Figure~\ref{fig:classInference}(f).
A solution to this constraint system gives us a solution to our data-flow
problem.

\begin{figure}[t!]
\centering
\includegraphics[width=0.8\linewidth]{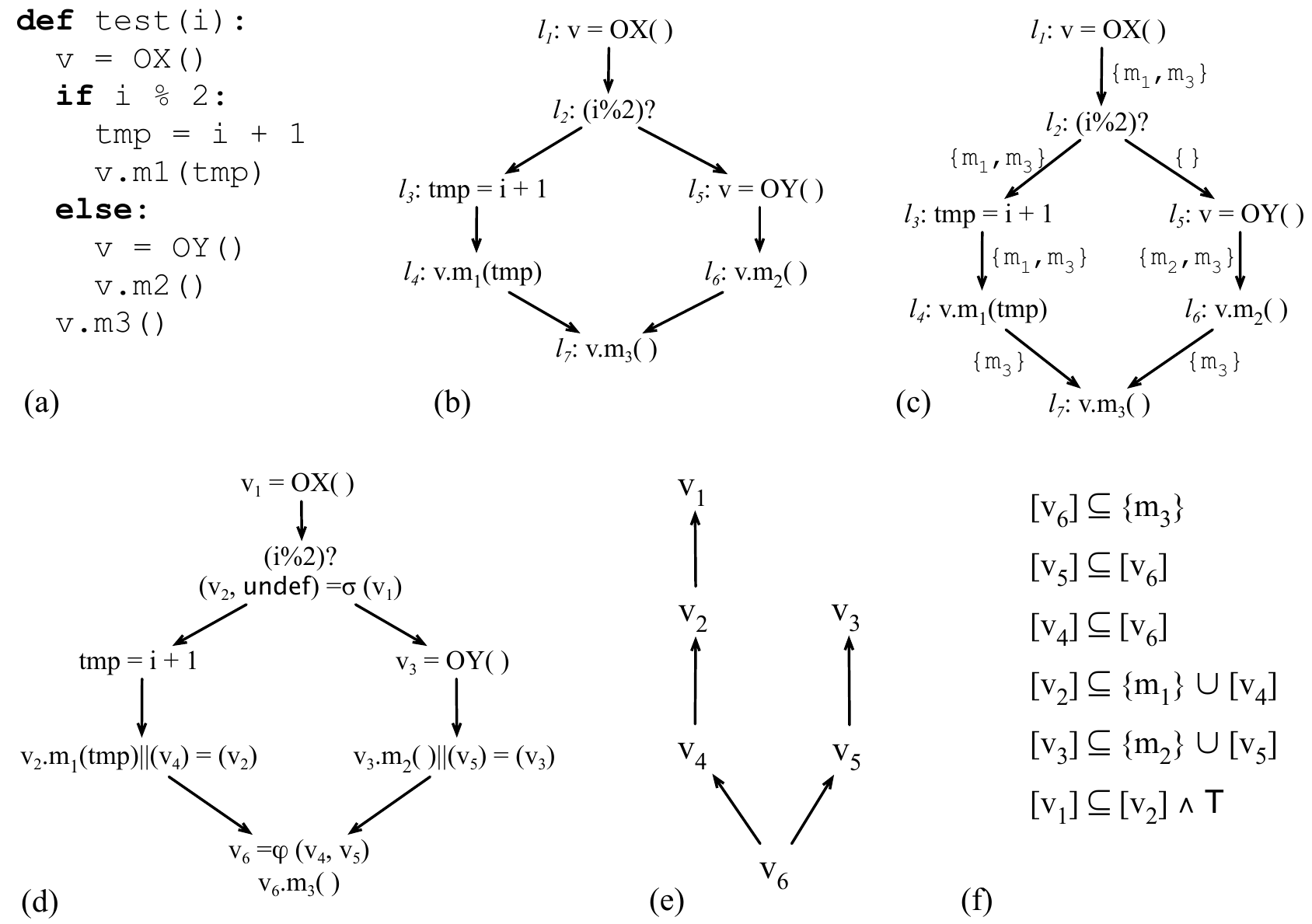}
\caption{\small Class inference as an example of backward data-flow analysis that takes information from the uses of variables.}
\label{fig:classInference}
\end{figure}

\begin{figure}[t!]
\centering
\includegraphics[width=0.8\linewidth]{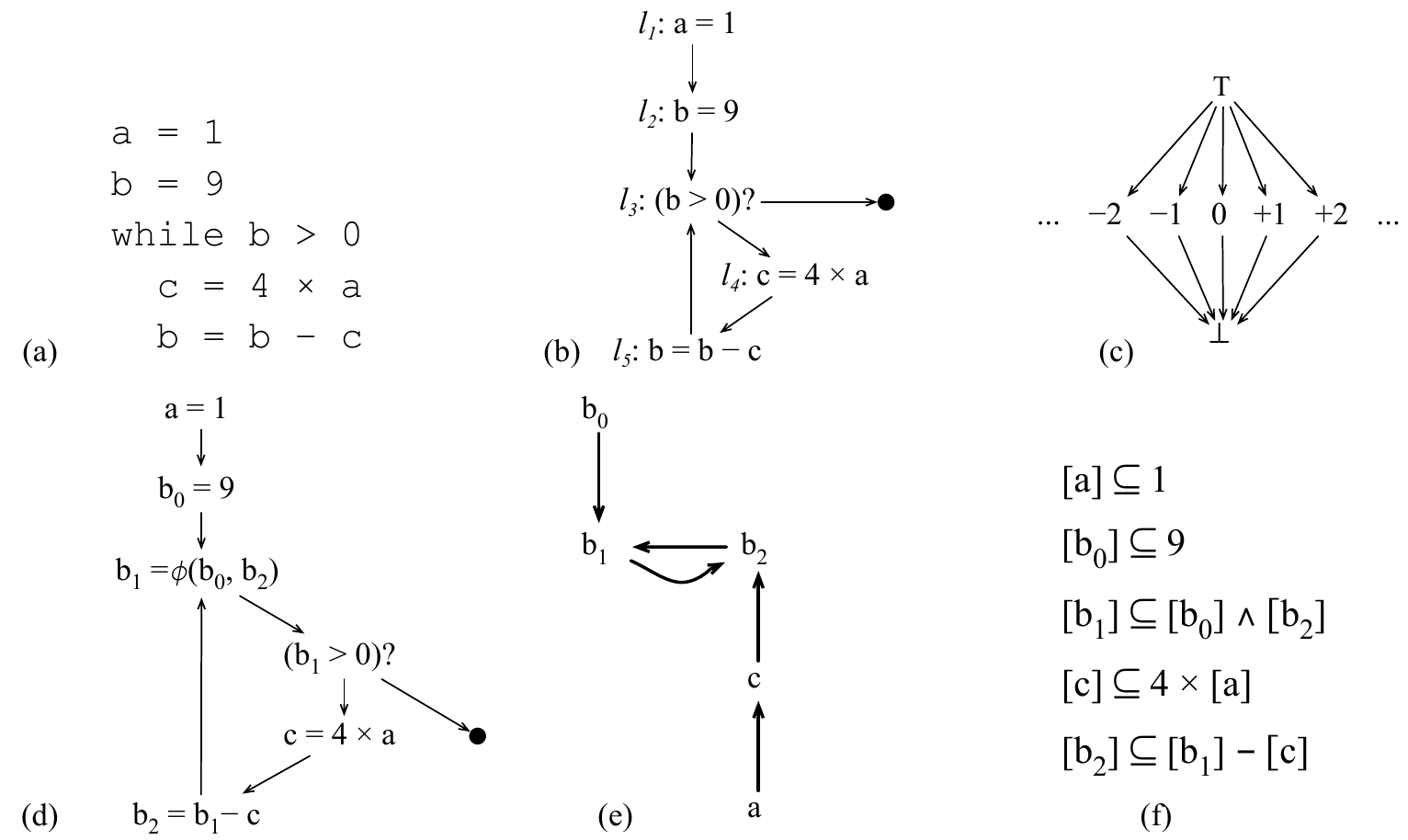}
\caption{\small Constant propagation as an example of forward data-flow analysis that takes information from the definitions of variables.}
\label{fig:constantProp}
\end{figure}

~\\
\noindent
\textbf{Constant Propagation: }
Figure~\ref{fig:constantProp} illustrates constant propagation, e.g.,
which variables in the program of
Figure~\ref{fig:constantProp}(a) can be replaced by constants?
The CFG of this program is given in Figure~\ref{fig:constantProp}(b).
Constant propagation has a very simple lattice $\cal L$, which we show in
Figure~\ref{fig:constantProp}(c).
In constant propagation, information is produced at the program points where
variables are defined.
Thus, in order to meet Definition~\ref{def:ssi}, we must guarantee that each
program point is reachable by a single definition of a variable.
Figure~\ref{fig:constantProp}(d) shows the intermediate representation that we create for the program in Figure~\ref{fig:constantProp}(b).
In this case, our intermediate representation is equivalent to the SSA form.
The def-use chains implicit in our program representation lead
to the constraint system shown in Figure~\ref{fig:constantProp}(f).
We can use the def-use chains seen in Figure~\ref{fig:constantProp}(e) to
guide a worklist-based constraint solver, as Nielson
{\em et al.}~\cite[Ch.6]{Nielson05} describe.

\begin{figure}[t!]
\centering
\includegraphics[width=0.8\linewidth]{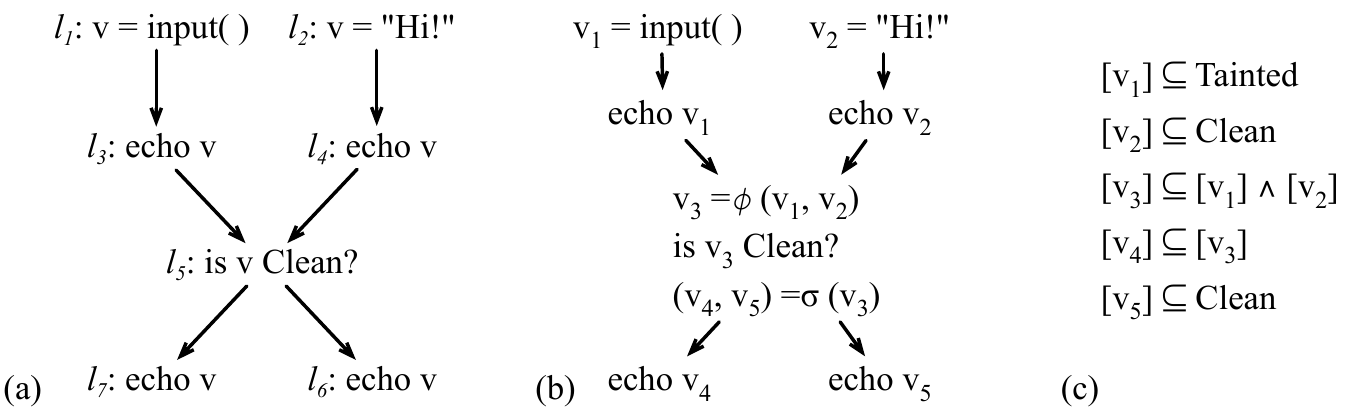}
\caption{
\label{fig:taintAnalysis}
\small Taint analysis is a forward data-flow analysis that takes information from the definitions of variables and conditional tests on these variables.
}
\end{figure}

~\\
\noindent
\textbf{Taint analysis:}
The objective of taint analysis~\cite{Rimsa11,Rimsa14} is to find program
vulnerabilities.
In this case, a harmful attack is possible when input data reaches sensitive
program sites without going through special functions called sanitizers.
Figure~\ref{fig:taintAnalysis} illustrates this type of analysis.
We have used $\phi$ and $\sigma$-functions to split the live ranges of the
variables in Figure~\ref{fig:taintAnalysis}(a) producing the program in
Figure~\ref{fig:taintAnalysis}(b).
Let us assume that {\em echo} is a sensitive function, because it is used to
generate web pages.
For instance, if the data passed to {\em echo} is a JavaScript program, then we
could have an instance of cross-site scripting attack.
Thus, the statement $\mathit{echo} \ v_1$ may be a source of vulnerabilities, as
it outputs data that comes directly from the program input.
On the other hand, we know that $\mathit{echo} \ v_2$ is always safe, for
variable $v_2$ is initialized with a constant value.
The call $\mathit{echo} \ v_5$ is always safe, because variable $v_5$ has been
sanitized; however, the call $\mathit{echo} \ v_4$ might be tainted, as variable
$v_4$ results from a failed attempt to sanitize $v$.
The def-use chains that we derive from the program representation lead
naturally to a constraint system, which we show in
Figure~\ref{fig:taintAnalysis}(c).
The intermediate representation that we create in this case is equivalent
to the {\em Extended Single Static Assignment} (e-SSA) form~\cite{Bodik00}.
It also suits the ABCD algorithm for array bounds-checking
elimination~\cite{Bodik00}, Su and Wagner's range analysis~\cite{Su05} and
Gawlitza {\em et al.}'s range analysis~\cite{Gawlitza09}.

\begin{figure}[t!]
\centering
\includegraphics[width=0.8\linewidth]{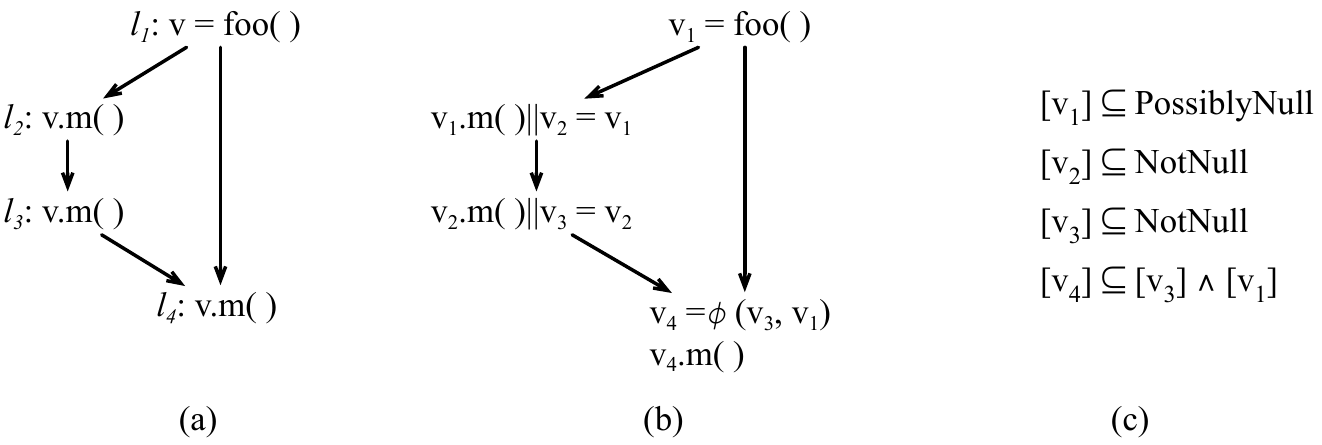}
\caption{\small Null pointer analysis as an example of forward data-flow analysis that
takes information from the definitions and uses of variables.}
\label{fig:nullAnalysis}
\end{figure}

~\\
\noindent
\textbf{Null pointer analysis: }
The objective of null pointer analysis is to determine which references may hold
null values.
Nanda and Sinha have used a variant of this analysis to find which method
dereferences may throw exceptions, and which may not~\cite{Nanda09}.
This analysis allows compilers to remove redundant null-exception tests and
helps developers to find null pointer dereferences.
Figure~\ref{fig:nullAnalysis} illustrates this analysis.
Because information is produced at use sites, we split live ranges after each
variable is used, as we show in Figure~\ref{fig:nullAnalysis}(b).
For instance, we know that the call $v_2.m()$ cannot result in a null pointer
dereference exception, otherwise an exception would have been thrown during the
invocation $v_1.m()$.
On the other hand, in Figure~\ref{fig:nullAnalysis}(c) we notice that the state
of $v_4$ is the meet of the state of $v_3$, definitely not-null, and the state
of $v_1$, possibly null, and we must conservatively assume that $v_4$ may be
null.

\section{Building the Intermediate Program Representation}
\label{sec:building}

A {\em live range splitting strategy} \
$\SS_v = I_\uparrow \cup I_\downarrow$ over a variable $v$ consists of two sets
of \splitpoints (see Section~\ref{sub:split} for a definition of \splitpoints).
We let $I_\downarrow$ denote a set of \splitpoints that produce
information for a forward analysis.
Similarly, we let $I_\uparrow$ denote a set of \splitpoints that are
interesting for a backward analysis.
The live-range of $v$ must be split at least at every \splitpoint in $\SS_v$.
Going back to the examples from Section~\ref{sub:examples}, we have the live
range splitting strategies enumerated below.
Further examples are given in Figure~\ref{fig:splittingSt}.

\begin{itemize}
\item \textbf{Class inference} is a backward analysis that takes information
from the uses of variables.
Thus, for each variable, the live-range splitting strategy contains the set
of \splitpoints where that variable is used.
For instance, in Figure~\ref{fig:classInference}(b), we have that
$\SS_{v} = \{l_4, l_6,l_7\}_\uparrow$.

\item \textbf{Constant propagation} is a forward analysis that takes
information from definition sites.
Thus, for each variable $v$, the live-range splitting strategy is
characterized by the set of points where $v$ is defined.
For instance, in Figure~\ref{fig:constantProp}(b), we have that
$\SS_{b} = \{l_2, l_5\}_\downarrow$.

\item \textbf{Taint analysis} is a forward analysis that takes information
from \splitpoints where variables are defined, and conditional tests that use these
variables.
For instance, in Figure~\ref{fig:taintAnalysis}(a), we have that
$\SS_{v} = \{l_1, l_2, \Out(l_5)\}_\downarrow$. 

\item Nanda {\em et al.}'s \textbf{null pointer analysis}~\cite{Nanda09} is a
forward flow problem that takes information from definitions and uses.
For instance, in Figure~\ref{fig:nullAnalysis}(a), we have that
$\SS_{v} = \{l_1, l_2, l_3, l_4\}_\downarrow$.
\end{itemize}

\begin{figure}[t!]
\begin{center}
\begin{small}
\renewcommand\arraystretch{1.4}
\begin{tabular}{| c | c |} \hline
{\bf Client} & {\bf Splitting strategy $\SS$} \\ \hline 
Alias analysis, reaching definitions & $\mathit{Defs}_\downarrow$ \\ 
cond. constant propagation~\cite{Wegman91} &  \\ \hline  
Partial Redundancy Elimination~\cite{Ananian99,Singer06} & $\mathit{Defs}_\downarrow \bigcup \mathit{LastUses}_\uparrow$ \\ \hline 
ABCD~\cite{Bodik00}, taint analysis~\cite{Rimsa11},  & $\mathit{Defs}_\downarrow \bigcup \mathit{\Out(Conds)}_\downarrow$ \\ 
range analysis~\cite{Su05,Gawlitza09} & \\ \hline 
Stephenson's bitwidth analysis~\cite{Stephenson00} & $\mathit{Defs}_\downarrow \bigcup \mathit{\Out(Conds)}_\downarrow \bigcup \mathit{Uses}_\uparrow$  \\ \hline 
Mahlke's bitwidth analysis~\cite{Mahlke01} & $\mathit{Defs}_\downarrow \bigcup \mathit{Uses}_\uparrow$  \\ \hline 
An's type inference~\cite{An11}, class inference~\cite{Chambers89} & $\mathit{Uses}_\uparrow$ \\ \hline 
Hochstadt's type inference~\cite{Hochstadt08} & $\mathit{Uses}_\uparrow \bigcup \mathit{\Out(Conds)}_\uparrow$ \\ \hline 
Null-pointer analysis~\cite{Nanda09} & $\mathit{Defs}_\downarrow \bigcup\mathit{Uses}_\downarrow$ \\ \hline
\end{tabular} \end{small} 
\caption{\small Live range splitting strategies for different data-flow analyses.
We use $\mathit{Defs}$ (resp. $\mathit{Uses}$) to denote the set of instructions that define (resp. use) the variable; $\mathit{Conds}$ to denote the set of instructions that apply a conditional test on a variable; $\Out(\mathit{Conds})$ the exits of the corresponding basic blocks; $\mathit{LastUses}$ to denote the set of instructions where a variable is used, and after which it is no longer live.}
\label{fig:splittingSt} \end{center} \end{figure}

\def\SSIfy{\textsf{SSIfy}}

\begin{figure}[htbp]
\begin{small}
\begin{tabular}{rl}
$_1$& \textsf{function \SSIfy}(var \var{v}, Splitting\_Strategy $\SS_v$)\\
$_2$& \1\textsf{split}($v$, $\SS_v$)\\
$_3$& \1\textsf{rename}($v$)\\
$_4$& \1\textsf{clean}($v$)\\
\end{tabular}
\end{small}
\caption{\label{fig:SSIfy} \small Split the live ranges of $v$ to convert it to SSI form}
\end{figure}

The algorithm \textsf{\SSIfy} in Figure~\ref{fig:SSIfy} implements a
live range splitting strategy in three steps:
\textsf{split}, \textsf{rename} and \textsf{clean}, which we describe in
the rest of this section.

~\\
\noindent
\textbf{Splitting live ranges through the creation of new definitions
of variables: }
To implement $\SS_v$, we must split the live ranges of $v$ at each
\splitpoint listed by $\SS_v$.
However, these \splitpoints are not the only ones where splitting might be
necessary.
As we have pointed out in Section~\ref{sub:split}, we might have, for the same original variable, many different sources of information reaching a common meet point.
For instance, in Figure~\ref{fig:taintAnalysis}(b), there exist two
definitions of variable $v$: $v_1$ and $v_2$, that reach the use of $v$ at $l_5$.
Information that flows forward from $l_3$ and $l_4$ collide at $l_5$,
the meet point of the if-then-else.
Hence the live-range of $v$ has to be split at the entry of $l_5$, e.g.,
at $\In(l_5)$, leading to a new definition $v_3$.
In general, the set of \splitpoints where information collide can be easily
characterized by join sets~\cite{Cytron91}.
The join set of a group of nodes $P$ contains the CFG nodes that can be
reached by two or more nodes of $P$ through disjoint paths.
Join sets  can be
over-approximated by the notion of iterated dominance frontier~\cite{Weiss92},
a core concept in SSA construction algorithms, which, for the sake of
completeness, we recall below:
\begin{itemize}
\item {\bf Dominance}:
a CFG node $n$ dominates a node $n'$ if every program path from the entry node
of the CFG to $n'$ goes across $n$.
If $n \neq n'$, then we say that $n$ {\em strictly} dominates $n'$.

\item {\bf Dominance frontier} ($DF$):
a node $n'$ is in the dominance frontier of a node $n$ if $n$ dominates a predecessor of $n'$, but does not strictly dominate $n'$.

\item {\bf Iterated dominance frontier} ($\mathit{DF}^+$):
the iterated dominance frontier of a node $n$ is the limit of the sequence:
\begin{eqnarray*}
\begin{array}{ccl}
DF_1 & = & DF(n) \\ DF_{i + 1} & = & DF_i \cup \{DF(z) \ | \ z \in DF_i\} \end{array}
\end{eqnarray*}
\end{itemize}
Similarly, split sets created by the backward propagation of information can
be over-approximated by the notion of {\em iterated post-dominance
frontier} ($\mathit{pDF^+}$), which is the 
$\mathit{DF^+}$~\cite{Appel02} of the CFG where orientation of edges have been reverted.
If $e = (u,v)$ is an edge in the control flow graph, then we
define the dominance frontier of $e$, i.e., $DF(e)$, as the dominance
frontier of a fictitious node $n$ placed at the middle of $e$.
In other words, $DF(e)$ is $DF(n)$, assuming that $(u,n)$ and
$(n,v)$ would exist.
Given this notion, we also define $DF^+(e)$, $pDF(e)$ and $pDF^+(e)$.

\begin{figure}[t!]
\begin{small}
\begin{tabular}{rl}
$_1$&{\sf function split}(var \var{v}, Splitting\_Strategy
$\SS_v = I_\downarrow \cup I_\uparrow$)\\
$_2$&\1 ``compute the set of split points"\\
$_3$&\1$S_\uparrow = \emptyset$\\
$_4$&\1\Foreach $i \in I_\uparrow$:\\
$_5$&\1\1 \If $i.\isjoin$:\\
$_6$&\1  \2 \Foreach $e\in \incoming(i)$:\\
$_7$&\1     \3  $S_\uparrow = S_\uparrow \bigcup \Out(pDF^+(e))$\\
$_8$&\1\1 \Else:\\
$_9$&\1  \2 $S_\uparrow = S_\uparrow \bigcup \Out(pDF^+(i))$\\
$_{10}$&\1$S_\downarrow = \emptyset$\\
$_{11}$&\1\Foreach $i \in S_\uparrow \bigcup \Defs(v) \bigcup I_\downarrow$:\\
$_{12}$&\1\1 \If $i.\isfork$:\\
$_{13}$&\1  \2 \Foreach $e \in \outgoing(i)$\\
$_{14}$&\1      \3 $S_\downarrow = S_\downarrow \bigcup \In(DF^+(e))$\\
$_{15}$&\1\1 \Else:\\
$_{16}$&\1  \2 $S_\downarrow = S_\downarrow \bigcup \In(DF^+(i))$\\
$_{17}$&\1$S = \SS_v \bigcup S_\uparrow \bigcup S_\downarrow$\\
$_{18}$&\1 ``Split live range of $v$ by inserting $\phi$, $\sigma$, and copies"\\
$_{19}$&\1\Foreach  $i \in S$:\\
$_{20}$&\1\1 \If $i$ does not already contain any definition of $v$:\\
$_{21}$&\1   \2  \If $i.\isjoin$: insert ``$v=\phi(v,...,v)$" at $i$\\
$_{22}$&\1   \2  \Else \If $i.\isfork$: insert ``$(v,...,v)= \sigma(v)$" at  $i$\\
$_{23}$&\1   \2 else: insert a copy ``$v=v$" at $i$\\
\end{tabular}
\caption{\label{fig:Spliting} \small Live range splitting. We use $\In(l)$ to denote a  \splitpoint at the entry of $l$, and $\Out(l)$ to denote a \splitpoint at the exit of $l$.
We let $\In(S) = \{\In(l)\ |\ l\in S\}$.
$\Out(S)$ is defined in a similar way.
}
\end{small}
\end{figure}

Figure~\ref{fig:Spliting} shows the algorithm that creates new definitions of
variables.
This algorithm has three phases.
First, in lines 3-9 we create new definitions to split the live ranges
of variables due to \emph{backward} collisions of information.
These new definitions are created at the iterated post-dominance
frontier of \splitpoints that originate information.
Notice that if the \splitpoint is a join (entry of a CFG node), 
information actually originate from each incoming edges (line 6).
In lines 10-16 we perform the inverse operation: we create new definitions of
variables due to the forward collision of information.
Finally, in lines 17-23 we actually insert the new definitions of $v$.
These new definitions might be created by $\sigma$ functions (due exclusively
to the splitting in lines 3-9); by $\phi$-functions (due exclusively to the
splitting in lines 10-16); or by parallel copies.
Contrary to Singer's algorithm, originally designed to produce SSI form
programs, we do not iterate between the insertion of $\phi$ and $\sigma$
functions.

The Algorithm {\sf split} preserves the SSA property, even for data-flow
analyses that do not require it.
As we see in line 11, the loop that splits meet nodes forwardly include, by
default, all the definition sites of a variable.
We chose to implement it in this way for practical reasons: the SSA
property gives us access to a fast liveness check~\cite{Benoit08}, which is
useful in actual compiler implementations.
This algorithm inserts $\phi$ and $\sigma$ functions conservatively.
Consequently, we may have these special instructions at \splitpoints that
are not true meet nodes.
In other words, we may have a $\phi$-function $v = \phi(v_1, v_2)$, in which
the abstract states of $v_1$ and $v_2$ are the same in a final solution of
the data-flow problem.

\begin{figure}[t!]
\begin{small}
\begin{tabular}{rl}
$_{1}$&{\sf function rename}(var $v$)\\
$_{2}$&\1 ``Compute use-def \& def-use chains"\\
$_{3}$&\1 ``We consider here that $\textit{stack}.\textrm{peek}()=\undef$ if
\textit{stack}.\textrm{isempty}(),\\
$_{4}$&\1~~~and that $\Def(\undef)=\textit{entry}$"\\
$_{5}$&\1$\textit{stack} = \emptyset$\\
$_{6}$&\1\Foreach CFG node $n$ in dominance order:\\
$_{7}$&\1\1 \Foreach $m$ that is a predecessor of $n$:\\
$_{8}$&\1     \2\If exists $d_m$ of the form ``$l^m:v=\dots$'' in a $\sigma$-function in $\Out(m)$:\\
$_{9}$&\1     \3 $\textit{stack}.\textrm{set\_def}(d_m)$\\
$_{10}$&\1  \2 \If exits $u_m$ of the form ``$\dots=l^m:v$'' in a $\phi$-function in $\In(n)$:\\
$_{11}$&\1     \3 $\textit{stack}.\textrm{set\_use}(u_m)$\\  
$_{12}$&\1\1 \If exists a $\phi$-function $d$ in $\In(n)$ that defines $v$:\\
$_{13}$&\1  \2 $\textit{stack}.\textrm{set\_def}(d)$\\
$_{14}$&\1\1 \Foreach instruction $u$ in $n$ that uses $v$:\\
$_{15}$&\1  \2 $\textit{stack}.\textrm{set\_use}(u)$\\
$_{16}$&\1\1 \If exists an instruction $d$ in $n$ that defines $v$:\\
$_{17}$&\1  \2 $\textit{stack}.\textrm{set\_def}(d)$\\
$_{18}$&\1\1 \Foreach $\sigma$-function $u$ in $\Out(n)$ that uses $v$:\\
$_{19}$&\1  \2 $\textit{stack}.\textrm{set\_use}(u)$\\
\end{tabular} \\ \\

\begin{tabular}{rl}
$_{21}$&{\sf function stack.set\_use}(instruction \var{inst}):\\
$_{22}$&\1\While $\Def(\textit{stack}.\textrm{peek()})$ does not dominate \var{inst}: \textit{stack}.\textrm{pop()}\\
$_{23}$&\1$v_i = \textit{stack}.\textrm{peek()}$\\
$_{24}$&\1replace the uses of $v$ by $v_i$ in \var{inst}\\
$_{25}$&\1\If $v_i\neq \undef$: set $\Uses(v_i)=\Uses(v_i) \bigcup inst$
\end{tabular} \\ \\

\begin{tabular}{rl}
$_{27}$&{\sf function stack.set\_def}(instruction \var{inst}):\\
$_{28}$&\1let $v_i$ be a fresh version of $v$\\
$_{29}$&\1replace the defs of $v$ by $v_i$ in \var{inst}\\
$_{30}$&\1set $\Def(v_i)= inst$\\
$_{31}$&\1$\textit{stack}.\textrm{push}(v_i)$
\end{tabular}
\end{small}
\caption{\label{fig:Rename} \small Versioning} 
\end{figure}

~\\
\noindent
\textbf{Variable Renaming: }
The algorithm in Figure~\ref{fig:Rename} builds def-use and use-def chains
for a program after live range splitting.
This algorithm is similar to the standard algorithm used to rename variables
during the SSA construction~\cite[Algorithm 19.7]{Appel02}.
To rename a variable $v$ we traverse the program's dominance tree, from top to
bottom, stacking each new definition of $v$ that we find.
The definition currently on the top of the stack is used to replace all the
uses of $v$ that we find during the traversal.
If the stack is empty, this means that the variable is not defined at that
point.
The renaming process replaces the uses of undefined variables by $\undef$
(line~3). 
We have two methods, \textit{stack}.\textrm{set\_use} and \textit{stack}.\textrm{set\_def} to build
the chain relations between the variables.
Notice that sometimes we must rename a single use inside a $\phi$-function,
as in lines~10-11 of the algorithm.
For simplicity we consider this single use as a simple
assignment when calling \textit{stack}.\textrm{set\_use}, as one can see in line 11.
Similarly, if we must rename a single definition inside a
$\sigma$-function, then we treat it as a simple assignment, like we do in
lines 8-9 of the algorithm.

\begin{figure}[t!]
\begin{small}
\begin{tabular}{rl}
$_{1}$ & \textsf{function clean}(var $v$)\\
$_{2}$ & \1 \Let \var{web} = $\{ v_i \ | \ v_i \textrm{ is a version of } v \}$\\
$_{3}$ & \1 \Let \var{defined} = $\emptyset$\\
$_{4}$ & \1 \Let \var{active} = \{ $\var{inst} \ | \ \var{inst}$ is actual instruction and $\var{web}\cap \var{inst}.\textrm{defs}  \neq \emptyset \}$\\
$_{5}$ & \1 \While exists {\em inst}  in {\em active} s.t. {\em web} $\cap$
{\em inst}.\textrm{defs} $\backslash$  {\em defined} $\neq \emptyset$:\\
$_{6}$ & \1  \1 \Foreach $v_i \in \var{web} \cap \var{inst}.\textrm{defs} \backslash \var{defined}$: \\
$_{7}$ & \1     \2 $\var{active} = \var{active} \cup \Uses(v_i)$ \\
$_{8}$ & \1     \2 $\var{defined} = \var{defined} \cup \{ v_i \}$ \\
$_{9}$ & \1 \Let $\var{used} = \emptyset$\\
$_{10}$ & \1 \Let $\var{active} = \{ \var{inst} \ | \var{inst}$ is actual instruction and $\var{web} \cap \var{inst}.\textrm{uses} \neq \emptyset \}$\\
$_{11}$ & \1 \While exists $\var{inst} \in \var{active}$ s.t.
$\var{inst}.\textrm{uses} \backslash  \var{used} \neq \emptyset$:\\
$_{12}$ & \1  \1 \Foreach $v_i \in \var{web} \cap \var{inst}.\textrm{uses} \backslash \textrm{used}$: \\
$_{13}$ & \1     \2 $\var{active} = \var{active} \cup \Def(v_i)$ \\
$_{14}$ & \1     \2 $\var{used} = \var{used} \cup \{ v_i \}$ \\
$_{15}$ & \1 \Let $\var{live} = \var{defined} \cap \var{used}$ \\
$_{16}$ & \1 \Foreach non actual $\var{inst} \in \Def(\var{web})$:\\
$_{17}$ & \1  \1 \Foreach $v_i$ operand of \var{inst} s.t. $v_i \notin
\var{live}$:\\
$_{18}$ & \1             \4 replace $v_i$ by $\undef$\\
$_{19}$ & \1  \1 \If $\var{inst}.\textrm{defs}=\{\undef\}$ or $\var{inst}.\textrm{uses}=\{\undef\}$\\
$_{20}$ & \1  \2 eliminate \var{inst} from the program\\
\end{tabular}
\end{small}
\caption{\label{fig:clean} \small Dead and undefined code elimination. Original instructions not inserted by \textsf{split} are called \emph{actual} instruction. We let {\em inst}.\textrm{defs} denote the set of variables defined by {\em inst}, and {\em inst}.\textrm{uses} denote the set of variables used by {\em inst}.}
\end{figure}

~\\
\noindent
\textbf{Dead and Undefined Code Elimination: }
The algorithm in Figure~\ref{fig:clean} eliminates $\phi$-functions that define variables not actually used in the code, $\sigma$-functions that use variables not actually defined in the code, and parallel copies that either define or use variables that do not reach any actual instruction.
``Actual'' instructions are those instructions that already existed in the program before we transformed it with {\sf split}.
In line~3 we let ``\var{web}'' be the set of versions of $v$, so as to restrict the cleaning process to variable~$v$, as we see in lines~4-6 and lines~10-12.
The set ``\var{active}'' is initialized to actual instructions in line~4.
Then, during the loop in lines~5-8 we add to active $\phi$-functions, $\sigma$-functions, and copies that can reach actual definitions through use-def chains.
The corresponding version of $v$ is then marked as \emph{defined} (line~8).
The next loop, in lines 11-14 performs a similar process to add to the active set the instructions that can reach actual uses through def-use chains.
The corresponding version of $v$ is then marked as \emph{used} (line~14).
Each non live variable (see line~15), i.e. either undefined or dead (non used) is replaced by $\undef$ in all $\phi$, $\sigma$, or copy functions where it appears.
This is done in lines~15-18.
Finally useless $\phi$, $\sigma$, or copy functions are removed in lines~19-20. 
As a historical curiosity, Cytron {\em et al.}'s procedure to build SSA form
produced what is called {\em the minimal representation}~\cite{Cytron91}.
Some of the $\phi$-functions in the minimal representation define variables
that are never used.
Briggs {\em et al.}~\cite{Briggs94} remove these variables; hence, producing what compiler writers normally call {\em pruned SSA-form}.
We close this section stating that the {\sf SSIfy} algorithm preserves the
semantics of the modified program~\footnote{The theorems in the main part of this paper
are proved in the appendix}:

\begin{theorem}[Semantics]
\label{theo:semantics}
{\sf SSIfy} maintains the following property:
if a value $n$ written into variable $v$ at \splitpoint $i'$ is read at a
\splitpoint $i$ in the original program, then the same value assigned to a
version of variable $v$ at \splitpoint $i'$ is read at a \splitpoint $i$
after transformation.
\end{theorem}

~\\
\noindent
\textbf{The Propagation Engine: }
Def-use chains can be used to solve, sparsely, a PLV problem about any program
that fulfills the SSI property.
However, in order to be able to rely on these def-use chains, we need to
derive a sparse constraint system from the original - dense - system.
This sparse system is constructed according to Definition~\ref{def:ssi_eq}.
Theorem~\ref{theo:ssify} states that such a system exists for any program,
and can be obtained directly from the Algorithm {\sf SSIfy}.
The algorithm in Figure~\ref{fig:propback} provides worklist based solvers for
backward and forward sparse data-flow systems built as in
Definition~\ref{def:ssi_eq}.

\begin{definition}[SSI constrained system]
\label{def:ssi_eq}
Let $E^{\textit{ssi}}_{\textit{dense}}$ be a forward (resp. backward) constraint system extracted from a program that meets
the SSI properties.
Hence, for each pair (variable $v$, \progpoint $p$) we have
equations $[v]^p = [v]^p \wedge F_v^{s,p}([v_1]^s, \dots, [v_n]^s)$.
We define a system of sparse equations $E^{\textit{ssi}}_{\textit{sparse}}$ as follows:
\begin{itemize}
\item Let $\{a, \ldots, b\}$ be the variables used (resp. defined) at \splitpoint $i$, where variable $v$ is defined (resp. used). Let $s$ and $p$ be the \progpoints around $i$.
The LINK property ensures that $F^{s,p}_v$ depends only on some
$[a]^s \dots [b]^s$.
Thus, there exists a function $G^i_v$ defined as the projection of
$F^{s,p}_v$ on ${\cal L}_a\times \dots \times{\cal L}_b$, such that
$G^i_v([a]^s, \dots, [b]^s) = F^{s,p}_v([v_1]^s,\dots, [v_n]^s)$.

\item The sparse constrained system associates with each variable $v$, and each 
definition (resp. use) point $i$ of $v$, the corresponding constraint
$[v]  \sqsubseteq G_v^i([a], \ldots, [b])$ where $a,\dots, b$ are used
(resp. defined) at $i$.
\end{itemize}
\end{definition}

\begin{theorem}[Correctness of SSIfy]
\label{theo:ssify}
The execution of~\textsf{\SSIfy}($v,\,\SS_v$), for every variable $v$ in the
target program, creates a new program representation such that:
\begin{enumerate}
\item there exists a system of equations $E^{\textit{ssi}}_{\textit{dense}}$, isomorphic to $E_{\textit{dense}}$ for which the new program representation fulfills the SSI property.
\item if $E_{\textit{dense}}$ is monotone then $E^{\textit{ssi}}_{\textit{dense}}$ is also monotone.
\end{enumerate}
\end{theorem}

\def\1{\qquad}
\def\2{\1\1}
\def\3{\2\1}
\def\4{\2\2}
\def\5{\3\2}
\def\6{\4\2}
\def\7{\5\2}
\def\8{\6\2}
\def\9{\7\2}
\def\If{{\sf  if }}
\def\Let{{\sf  let }}
\def\Then{{\sf  then }}
\def\Else{{\sf  else}}
\def\Foreach{{\sf foreach }}
\def\For{{\sf for }}
\def\While{{\sf while }}

\newcommand\val[1]{[#1]}
\begin{figure}[t!]
\begin{small}
\begin{tabular}{rl}
$_1$ & \textsf{function forward\_propagate}(transfer\_functions $\cal G$)\\
$_2$ & \1$\var{worklist} = \emptyset$\\
$_3$ & \1\Foreach variable $v$: $\val{v}=\top$\\
$_4$ & \1\Foreach instruction $i$: $\var{worklist}\ +\hspace{-0.3em}= i$\\
$_5$ & \1\While $\var{worklist}\neq \emptyset$:\\
$_6$ & \1\1 \Let $i \in \var{worklist}$ \\
$_7$ & \1\1 $\var{worklist}\ -\hspace{-0.3em}= i$\\
  $_8$ & \1\1 \Foreach $v \in i.\textrm{defs}$:\\
$_9$ & \1  \2  $\val{v}_{new} = \val{v} \wedge G_v^i(\val{i.\textrm{uses}})$\\
$_{10}$&  \1 \2  \If $\val{v} \neq \val{v}_{new}$: \\
$_{11}$& \1   \3 $\var{worklist}\ +\hspace{-0.3em}=\Uses(v)$\\
$_{12}$& \1   \3 $\val{v} = \val{v}_{new}$\\
\end{tabular}
\end{small}
\caption{\label{fig:propback} \small Forward propagation engine under SSI.
For backward propagation, we replace $i$.defs by $i$.uses, $i$.uses by $i$.defs, and $\Uses(v)$ by $\Def(v)$}
\end{figure}

\section{Our Approach vs Other Sparse Evaluation Frameworks}
\label{sec:rel}

There have been previous efforts to provide theoretical and practical
frameworks in which data-flow analyses could be performed sparsely.
In order to clarify some details of our contribution, this section compares it
with three previous approaches:
Choi's Sparse Evaluation Graphs, Ananian's Static Single Information form and
Oh's Sparse Abstract Interpretation Framework.

~\\
\noindent
\textbf{Sparse Evaluation Graphs: }
Choi's {\em Sparse Evaluation Graphs}~\cite{Choi91} are one of the earliest
data-structures designed to support sparse analyses.
The nodes of this graph represent program regions where information produced by
the data-flow analysis might change.
Choi {\em et al.}'s ideas have been further expanded, for example, by Johnson
{\em et al.}'s {\em Quick Propagation Graphs}~\cite{Johnson93}, or Ramalingan's
{\em Compact Evaluation Graphs}~\cite{Ramalingan02}.
Nowadays we have efficient algorithms that build such
data-structures~\cite{Johnson94,Pingali97}.
These graphs improve many data-flow analyses in terms of runtime and memory
consumption.
However, they are more limited than our approach, because they
can only handle sparsely problems that Zadeck has classified as
{\em Partitioned Variable} (PVP).
In these problems, a program variable can be analyzed independently from the
others.
Reaching definitions and liveness analysis are examples of PVPs, as this kind
of information can be computed for one program variable independently from the
others.
For these problems we can build intermediate program representations isomorphic
to SEGs, as we state in Theorem~\ref{theo:iso}.
However, many data-flow problems, in particular the PLV analyses that we
mentioned in Section~\ref{sub:examples}, do not fit into this category.
Nevertheless, we can handle them sparsely.
The SEGs can still support PLV problems, but, in this case,
a new SEG vertex would be created for every \splitpoint where new information
is produced, and we would have a dense analysis.

\begin{theorem}[Equivalence SSI/SEG]
\label{theo:iso}
Given a forward Sparse Evaluation Graph (\textit{SEG}) that represents
a variable $v$ in a program representation \textit{Prog} with CFG $G$, there exists a
live range splitting strategy that once applied on $v$ builds a program
representation that is isomorphic to \textit{SEG}.
\end{theorem}

~\\
\noindent
\textbf{Static Single Information Form and Similar Program Representations: }
Scott Ananian has introduced in the late nineties the {\em Static Single
Information} (SSI) form, a program representation that supports both forward and
backward analyses~\cite{Ananian99}.
This representation was later revisited by Jeremy Singer~\cite{Singer06}.
The $\sigma$-functions that we use in this paper is a notation borrowed from
Ananian's work, and the algorithms that we discuss in
Section~\ref{sec:building} improve on Singer's ideas.
Contrary to Singer's algorithm we do not iterate between the insertion of
phi and sigma functions.
Consequently, as we will show in Section~\ref{sec:exp}, we insert less
phi and sigma functions.
Nonetheless, as we show in Theorem~\ref{theo:ssify}, our method is enough to
ensure the SSI properties for any combination of unidirectional problems.
In addition to the SSI form, we can emulate several other different representations, by changing our parameterizations. Notice that for SSI we have $\{\mathit{Defs}_\downarrow \cup \mathit{LastUses}_\uparrow\}$.
For Bodik's e-SSA~\cite{Bodik00} we have $\mathit{Defs}_\downarrow \bigcup \mathit{\Out(\textit{Conds})}_\downarrow$.
Finally, for SSU~\cite{George03,Lo98,Plevyak96} we have $\mathit{Uses}_\uparrow$.

The SSI constrained system might have several inequations for the same
left-hand-side, due to the way we insert phi and sigma functions.
Definition~\ref{def:ssi}, as opposed to the original SSI
definition~\cite{Ananian99,Singer06}, does not ensure the SSA or the SSU
properties.
These guarantees are not necessary to every sparse analysis.
It is a common assumption in the compiler's literature that ``data-flow analysis
(\dots) can be made simpler when each variable has only one definition", as
stated in Chapter 19 of Appel's textbook~\cite{Appel02}.
A naive interpretation of the above statement could lead one to conclude that
data-flow analyses become simpler as soon as the program representation enforces
a single source of information per live-range: SSA for forward propagation, SSU
for backward, and the \emph{original} SSI for bi-directional analyses.
This premature conclusion is contradicted by the example of dead-code
elimination, a backward data-flow analysis that the SSA form simplifies.
Indeed, the SSA form fulfills our definition of the SSI property for
dead-code elimination.
Nevertheless, the corresponding constraint system may have several
inequations, with the same left-hand-side, i.e., one for each use of a given
variable $v$.
Even though we may have several sources of information, we can still solve
this backward analysis using the algorithm in Figure~\ref{fig:propback}.
To see this fact, we can replace $G_v^i$ in Figure~\ref{fig:propback} by
``\emph{i is a useful instruction or one of its definitions is marked as
useful}'' and one obtains the classical algorithm for dead-code elimination.

~\\
\noindent
\textbf{Sparse Abstract Interpretation Framework: }
Recently, Oh {\em et al.}~\cite{Oh12} have designed and tested a framework that
sparsifies flow analyses modelled via abstract interpretation.
They have used this framework to implement standard analyses on the
interval~\cite{Cousot77} and on the octogon lattices~\cite{Mine06}, and have
processed large code bodies.
We believe that our approach leads to a sparser implementation.
We base this assumption on the fact that Oh {\em et al.}'s approach relies on
standard def-use chains to propagate information, whereas in our case, the
merging nodes combine information before passing it ahead.
As an example, lets consider the code
\texttt{if () then a=$\bullet$; else a=$\bullet$; endif if () then $\bullet$=a;
else $\bullet$=a; endif} under a forward analysis that generates information at
definitions and requires it at uses.
We let the symbol $\bullet$ denote unimportant values.
In this scenario, Oh et al.'s framework creates four dependence links between
the two \splitpoints where information is produced and the two \splitpoints where it is
consumed.
Our method, on the other hand, converts the program to SSA form; hence,
creating two names for variable a.
We avoid the extra links because a $\phi$-function merges the data that comes
from these names before propagating it to the use sites.

\section{Experimental Results}
\label{sec:exp}


This section describes an empirical evaluation of the size and runtime
efficiency of our algorithms.
Our experiments were conducted on a dual core \texttt{Intel Pentium D} of 2.80GHz of clock, 1GB of memory, running \texttt{Linux Gentoo}, version 2.6.27.
Our framework runs in LLVM 2.5~\cite{Lattner04}, and it passes all the tests that LLVM does.
The LLVM test suite consists of over 1.3 million lines of C code.
In this paper we show results for SPEC CPU 2000.
To compare different live range splitting strategies we generate the program
representations below.
Figure~\ref{fig:splittingSt} explains the sets {\em Defs, Uses} and {\em Conds}.
\begin{enumerate}

\item \textit{SSI}:
Ananian's Static Single Information form~\cite{Ananian99} is our baseline.
We build the SSI program representation via Singer's iterative algorithm.

\item \textit{ABCD}:
$(\{\mathit{Defs}, \mathit{Conds}\}_\downarrow)$.
This live range splitting strategy generalizes the ABCD algorithm for array bounds checking elimination~\cite{Bodik00}.
An example of this live range splitting strategy is given in Figure~\ref{fig:taintAnalysis}.

\item \textit{CCP}:
$(\{\mathit{Defs}, \mathit{Conds}_{eq}\}_\downarrow)$.
This splitting strategy, which supports Wegman {\em et al.}'s~\cite{Wegman91} conditional constant propagation, is a subset of the previous strategy.
Differently of the ABCD client, this client requires that only variables used in equality tests, e.g., \texttt{==}, undergo live range splitting.
That is, $\mathit{Conds}_{eq}(v)$ denotes the conditional tests that check if $v$ equals a given value.

\end{enumerate}

\begin{figure}[t!]
\begin{center}
\includegraphics[width=0.8\linewidth]{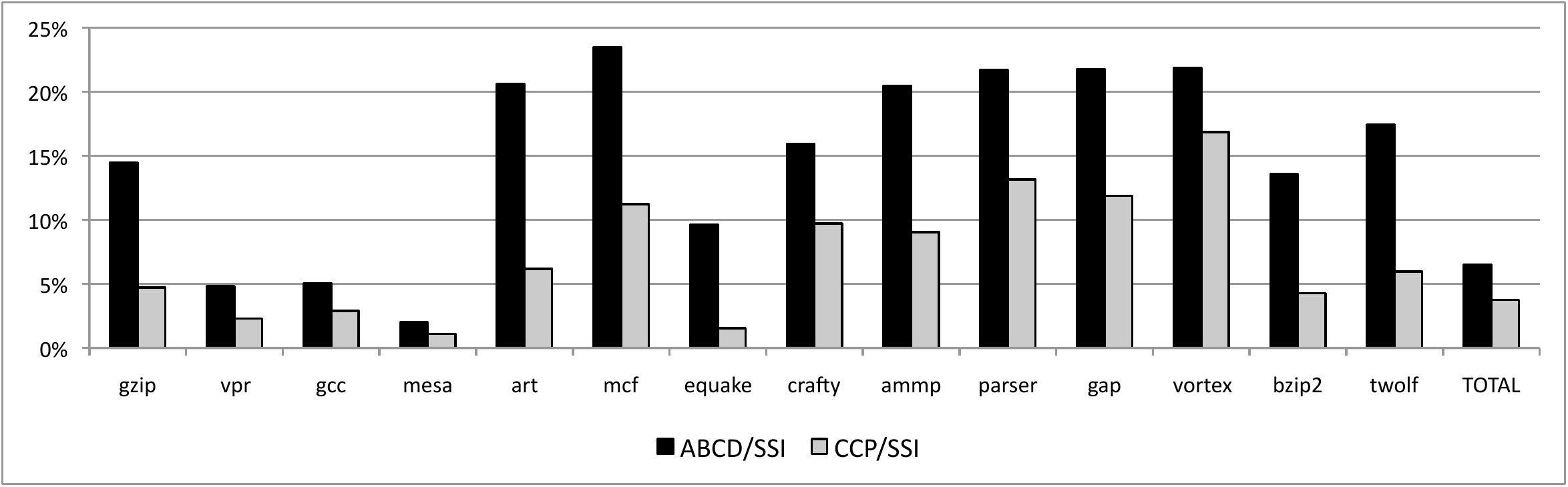} \caption{\small Comparison of the time taken to produce the different representations.
100\% is the time to use the SSI live range splitting strategy.
The shorter the bar, the faster the live range splitting strategy.
The SSI conversion took 1315.2s in total, the ABCD conversion took 85.2s, and the CCP conversion took 49.4s.
} \label{fig:partialXfull} \end{center} \end{figure}

~\\
\noindent
\textbf{Runtime: }
The chart in Figure~\ref{fig:partialXfull} compares the execution time of the three live range splitting strategies.
We show only the time to perform live range splitting.
The time to execute the optimization itself, removing array bound checks or performing constant propagation, is not shown.
The bars are normalized to the running time of the SSI live range splitting strategy.
On the average, the ABCD client runs in 6.8\% and the CCP client runs in 4.1\% of the time of SSI.
These two forward analyses tend to run faster in benchmarks with sparse control flow graphs, which present fewer conditional branches, and therefore fewer opportunities to restrict the ranges of variables.

In order to put the time reported in Figure~\ref{fig:partialXfull} in perspective, Figure~\ref{fig:partialXLLVM} compares the running time of our live range splitting algorithms with the time to run the other standard optimizations in our baseline compiler\footnote{To check the list of LLVM's target independent optimizations try \texttt{llvm-as < /dev/null | opt -std-compile-opts -disable-output -debug-pass=Arguments}}.
In our setting, LLVM -O1 runs 67 passes, among analysis and optimizations, which include partial redundancy elimination, constant propagation, dead code elimination, global value numbering and invariant code motion.
We believe that this list of passes is a meaningful representative of the optimizations that are likely to be found in an industrial strength compiler.
The bars are normalized to the optimizer's time, which consists of the time taken by machine independent optimizations plus the time taken by one of the live range splitting clients, e.g, ABCD or CCP.
The ABCD client takes 1.48\% of the optimizer's time, and the CCP client takes 0.9\%.
To emphasize the speed of these passes, we notice that the bars do not include the time to do machine dependent optimizations such as register allocation.

\begin{figure}[t!]
\begin{center}
\includegraphics[width=0.8\linewidth]{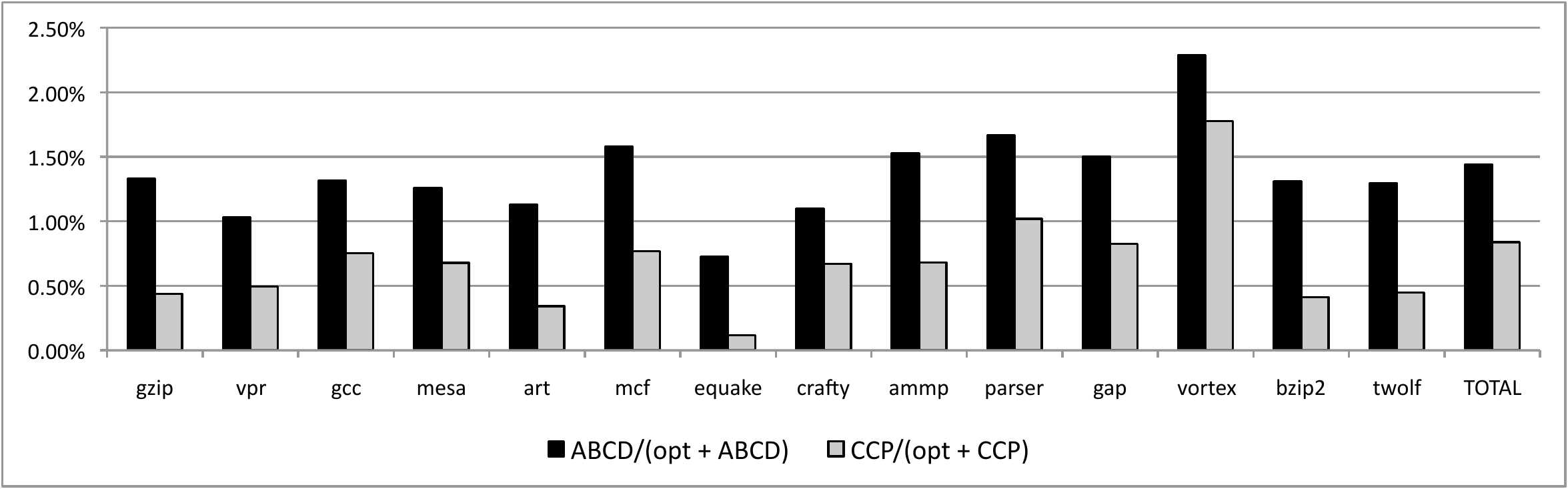} \caption{\small Execution time of two different live range splitting strategies compared to
the total time taken by machine independent LLVM optimizations
(\texttt{opt}  -O1).
100\% is the time taken by \texttt{opt}.
The shorter the bar, the faster the conversion.
} \label{fig:partialXLLVM} \end{center} \end{figure}

~\\
\noindent
\textbf{Space: }
Figure~\ref{fig:ProgGrowth} outlines how much each live range splitting strategy increases program size. We show results only to the ABCD and CCP clients, to keep the chart easy to read.
The SSI conversion increases program size in 17.6\% on average.
This is an absolute value, i.e., we sum up every $\phi$ and $\sigma$ function inserted, and divide it by the number of bytecode instructions in the original program.
This compiler already uses the SSA-form by default, and we do not count as new instructions the $\phi$-functions originally used in the program.
The ABCD client increases program size by 2.75\%, and the CCP client increases program size by 1.84\%.

\begin{figure}[t!]
\begin{center}
\includegraphics[width=0.8\linewidth]{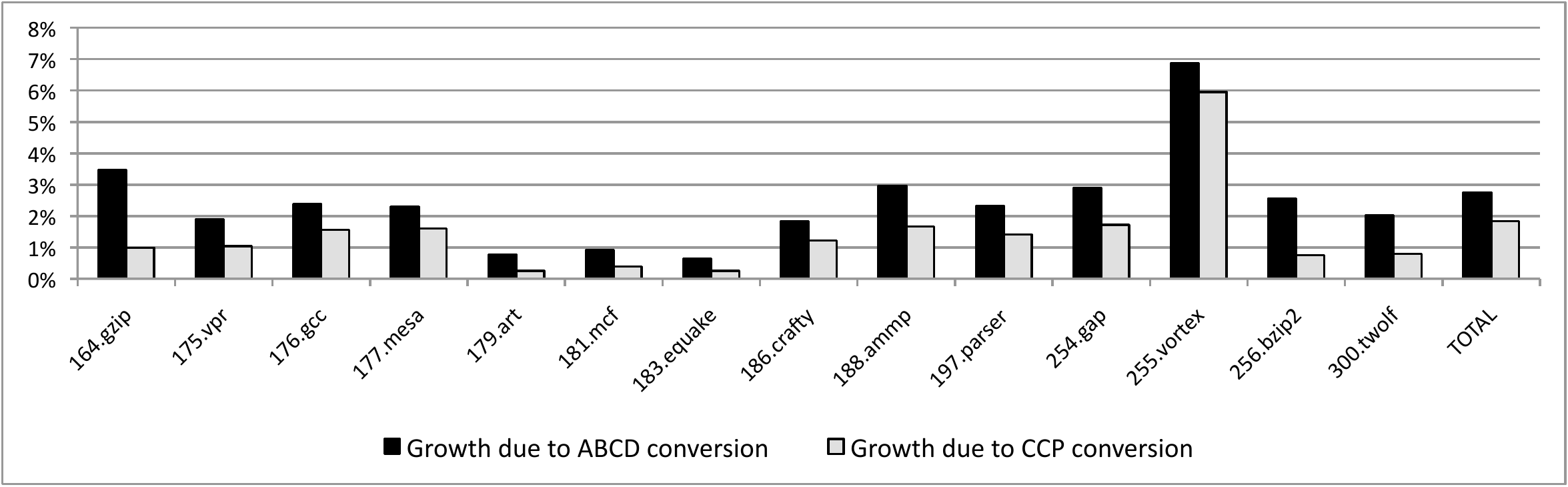} \caption{\small Growth in program size due to the insertion of new $\phi$ and $\sigma$ functions to perform live range splitting.} \label{fig:ProgGrowth} \end{center} \end{figure}

An interesting question that deserves attention is ``What is the benefit of using a sparse data-flow analysis in practice?" We have not implemented dense versions of the ABCD or the CCP clients.
However, previous works have shown that sparse analyses tend to outperform equivalent dense versions in terms of time and space efficiency~\cite{Choi91,Ramalingan02}.
In particular, the e-SSA format used by the ABCD and the CCP optimizations is the same program representation adopted by the tainted flow framework of Rimsa {\em et al.}~\cite{Rimsa11,Rimsa14}, which has been shown to be faster than a dense implementation of the analysis, even taking the time to perform live range splitting into consideration.

\section{Conclusion}
\label{sec:con}
This paper has presented a systematic way to build program representations that suit sparse data-flow analyses.
We build different program representations by splitting the live ranges of variables.
The way in which we split live ranges depends on two factors:
(i) which \splitpoints produce new information, e.g., uses, definitions,
tests, etc; and (ii), how this information propagates along the variable live
range: forwardly or backwardly.
We have used an implementation of our framework in LLVM to convert programs to the Static Single Information form~\cite{Ananian99}, and to provide intermediate representations to the ABCD array bounds-check elimination algorithm~\cite{Bodik00} and to Wegman {\em et al.}'s Conditional Constant Propagation algorithm~\cite{Wegman91}.
Our framework has been used by Couto {\em et al.}~\cite{Couto11} and by
Rodrigues {\em et al.}~\cite{Rodrigues13} in different implementations of
range analyses.
We have also used our live range splitting algorithm, implemented in the \texttt{phc} PHP compiler~\cite{Biggar09b,Biggar09}, to provide the Extended
Static Single Assignment form necessary to solve the tainted flow
problem~\cite{Rimsa11,Rimsa14}.

~\\
\noindent
\textbf{Extending our Approach. }
For the sake of simplicity, in this paper we have restricted our
discussion to: non relational analysis (PLV), intermediate-representation based appoach, and scalar variables without aliasing.
 
(1) \emph{non relation analysis}. In this paper we have focused on PLV
problems, i.e. solved by analyses that associate some information with each variable individually. For instance, we bind $i$ to a range
$0\leq i < \texttt{MAX\_N}$, but we do not relate $i$ and $j$, as in $0\leq i< j$. A relational analysis that provides a all-to-all relation between all variables of the program is dense by nature, as any \splitpoint both produces and consumes information for the analysis. Nevertheless, our framework is compatible with the notion of {\em packing}. Each pack is a set of variable groups selected to be related together. This approach is usually adopted in practical relational analyses, such as those used in Astr\'{e}e~\cite{Cousot09,Mine06}.

(2) \emph{IR based approach}. Our framework constructs an intermediate representation (IR) that preserves the semantic of the program. Like the SSA form, this IR has to be updated, and prior to final code generation, destructed. Our own experience as compiler developers let us believe that manipulating an IR such as SSA has many engineering advantages over building, and afterward dropping, a separate sparse evaluation graph (SEG) for each analysis. Testimony of this observation is the fact that the SSA form is used in virtually every modern compiler. Although this opinion is admittedly arguable, we would like to point out that updating and destructing our SSI form is equivalent to the update and destruction of SSA form. More importantly, there is no fundamental limitation in using our technique to build a separate SEG without modifying the IR. This SEG will inherit the sparse properties as his corresponding SSI flavor, with the benefit of avoiding the quadratic complexity of direct def-use chains ($|\textrm{Defs}(v)|\times|\textrm{Uses}(v)|$ for a variable $v$) thanks to the use of $\phi$ and $\sigma$ nodes. Note that this quadratic complexity becomes critical when dealing with code with aliasing or predication~\cite[pp.234]{Oh12}.

(3) \emph{analysis of scalar variables without aliasing or predication}. The most successful flavor of SSA form is the minimal and pruned representation restricted to scalar variables. The SSI form that we describe in this paper is akin to this flavor. Nevertheless, there exists several extensions to deal with code with predication (e.g. $\psi$-SSA form~\cite{Ferriere07}) and aliasing (e.g. Hashed~SSA~\cite{chow:hssa} or Array~SSA~\cite{FiKS00}). Such extensions can be applied without limitations to our SSI form allowing a wider range of analyses involving object aliasing and predication.

\begin{small}
\bibliographystyle{plain}
\bibliography{references}

\begin{thebibliography}{10}

\bibitem{Ackerman84}
W.~B. Ackerman.
\newblock {\em Efficient Implementation of Applicative Languages}.
\newblock PhD thesis, MIT, 1984.

\bibitem{Ananian99}
Scott Ananian.
\newblock The static single information form.
\newblock Master's thesis, {MIT}, September 1999.

\bibitem{Appel02}
Andrew~W. Appel and Jens Palsberg.
\newblock {\em Modern Compiler Implementation in Java}.
\newblock Cambridge University Press, 2nd edition, 2002.

\bibitem{Biggar09b}
Paul Biggar.
\newblock {\em Design and Implementation of an Ahead-of-Time Compiler for
  {PHP}}.
\newblock PhD thesis, Trinity College Dublin, 2009.

\bibitem{Biggar09}
Paul Biggar, Edsko de~Vries, and David Gregg.
\newblock A practical solution for scripting language compilers.
\newblock In {\em SAC}, pages 1916--1923. ACM, 2009.

\bibitem{Bodik00}
Rastislav Bodik, Rajiv Gupta, and Vivek Sarkar.
\newblock {ABCD}: eliminating array bounds checks on demand.
\newblock In {\em PLDI}, pages 321--333. ACM, 2000.

\bibitem{Benoit08}
Benoit Boissinot, Sebastian Hack, Daniel Grund, Benoit~Dupont de~Dinechin, and
  Fabrice Rastello.
\newblock Fast liveness checking for {SSA}-form programs.
\newblock In {\em CGO}, pages 35--44. IEEE, 2008.

\bibitem{Briggs94}
Preston Briggs, Keith~D. Cooper, and Linda Torczon.
\newblock Improvements to graph coloring register allocation.
\newblock {\em TOPLAS}, 16(3):428--455, 1994.

\bibitem{Budimlic02}
Zoran Budimlic, Keith~D. Cooper, Timothy~J. Harvey, Ken Kennedy, Timothy~S.
  Oberg, and Steven~W. Reeves.
\newblock Fast copy coalescing and live-range identification.
\newblock In {\em PLDI}, pages 25--32. ACM, 2002.

\bibitem{Cartwright89}
Robert Cartwright and Mattias Felleisen.
\newblock The semantics of program dependence.
\newblock {\em SIGPLAN Not.}, 24(7):13--27, 1989.

\bibitem{Chambers89}
Craig Chambers and David Ungar.
\newblock Customization: optimizing compiler technology for self, a
  dynamically-typed object-oriented programming language.
\newblock {\em SIGPLAN Not.}, 24(7):146--160, 1989.

\bibitem{Choi91}
Jong-Deok Choi, Ron Cytron, and Jeanne Ferrante.
\newblock Automatic construction of sparse data flow evaluation graphs.
\newblock In {\em POPL}, pages 55--66. ACM, 1991.

\bibitem{chow:hssa}
Fred Chow, Sun Chan, Shin-Ming Liu, Raymond Lo, and Mark Streich.
\newblock Effective representation of aliases and indirect memory operations in
  {SSA} form.
\newblock In {\em {CC}}, pages 253--267. Springer, 1996.

\bibitem{Cousot77}
P.~Cousot and R.~Cousot.
\newblock Abstract interpretation: a unified lattice model for static analysis
  of programs by construction or approximation of fixpoints.
\newblock In {\em POPL}, pages 238--252. ACM, 1977.

\bibitem{Cousot09}
Patrick Cousot, Radhia Cousot, J{\'e}r\^{o}me Feret, Laurent Mauborgne, Antoine
  Min{\'e}, and Xavier Rival.
\newblock Why does astr{\'e}e scale up?
\newblock {\em Form. Methods Syst. Des.}, 35(3):229--264, 2009.

\bibitem{Cytron91}
Ron Cytron, Jeanne Ferrante, Barry~K. Rosen, Mark~N. Wegman, and F.~Kenneth
  Zadeck.
\newblock Efficiently computing static single assignment form and the control
  dependence graph.
\newblock {\em TOPLAS}, 13(4):451--490, 1991.

\bibitem{Damas82}
Luis Damas and Robin Milner.
\newblock Principal type-schemes for functional programs.
\newblock In {\em POPL}, pages 207--212, New York, NY, USA, 1982. ACM.

\bibitem{Ferriere07}
Fran\c{c}ois de~Ferri\`{e}re.
\newblock Improvements to the $\psi$-{SSA} representation.
\newblock In {\em SCOPES}, pages 111--121. ACM, 2007.

\bibitem{Couto11}
Douglas do~Couto~Teixeira and Fernando Magno~Quintao Pereira.
\newblock The design and implementation of a non-iterative range analysis
  algorithm on a production compiler.
\newblock In {\em SBLP}, pages 45--59. SBC, 2011.

\bibitem{FiKS00}
Stephen~J. Fink, Kathleen Knobe, and Vivek Sarkar.
\newblock Unified analysis of array and object references in strongly typed
  languages.
\newblock In {\em SAS}, pages 155--174. Springer, 2000.

\bibitem{Gawlitza09}
Thomas Gawlitza, Jerome Leroux, Jan Reineke, Helmut Seidl, Gregoire Sutre, and
  Reinhard Wilhelm.
\newblock Polynomial precise interval analysis revisited.
\newblock {\em Efficient Algorithms}, 1:422 -- 437, 2009.

\bibitem{George03}
Lal George and Blu Matthias.
\newblock Taming the {IXP} network processor.
\newblock In {\em PLDI}, pages 26--37. ACM, 2003.

\bibitem{An11}
Jong hoon An, Avik Chaudhuri, Jeffrey~S. Foster, and Michael Hicks.
\newblock Dynamic inference of static types for ruby.
\newblock In {\em POPL}, pages 459--472. ACM, 2011.

\bibitem{Johnson94}
R.~Johnson, D.~Pearson, and K.~Pingali.
\newblock The program tree structure.
\newblock In {\em PLDI}, pages 171--185. ACM, 1994.

\bibitem{Johnson93}
Richard Johnson and Keshav Pingali.
\newblock Dependence-based program analysis.
\newblock In {\em PLDI}, pages 78--89. ACM, 1993.

\bibitem{Lattner04}
Chris Lattner and Vikram~S. Adve.
\newblock {LLVM}: A compilation framework for lifelong program analysis {\&}
  transformation.
\newblock In {\em CGO}, pages 75--88. IEEE, 2004.

\bibitem{Lo98}
Raymond Lo, Fred Chow, Robert Kennedy, Shin-Ming Liu, and Peng Tu.
\newblock Register promotion by sparse partial redundancy elimination of loads
  and stores.
\newblock In {\em PLDI}, pages 26--37. ACM, 1998.

\bibitem{Mahlke01}
S.~Mahlke, R.~Ravindran, M.~Schlansker, R.~Schreiber, and T.~Sherwood.
\newblock Bitwidth cognizant architecture synthesis of custom hardware
  accelerators.
\newblock {\em TCAD}, 20(11):1355--1371, 2001.

\bibitem{Mine06}
Antoine Min\'{e}.
\newblock The octagon abstract domain.
\newblock {\em Higher Order Symbol. Comput.}, 19:31--100, 2006.

\bibitem{Nanda09}
Mangala~Gowri Nanda and Saurabh Sinha.
\newblock Accurate interprocedural null-dereference analysis for java.
\newblock In {\em ICSE}, pages 133--143, 2009.

\bibitem{Nielson05}
Flemming Nielson, Hanne~Riis Nielson, and Chris Hankin.
\newblock {\em Principles of program analysis}.
\newblock Springer, 2005.

\bibitem{Oh12}
Hakjoo Oh, Kihong Heo, Wonchan Lee, Woosuk Lee, and Kwangkeun Yi.
\newblock Design and implementation of sparse global analyses for c-like
  languages.
\newblock In {\em {PLDI}}, pages 229--238. ACM, 2012.

\bibitem{Pingali97}
Keshav Pingali and Gianfranco Bilardi.
\newblock Optimal control dependence computation and the roman chariots
  problem.
\newblock In {\em TOPLAS}, pages 462--491. ACM, 1997.

\bibitem{Plevyak96}
John~Bradley Plevyak.
\newblock {\em Optimization of Object-Oriented and Concurrent Programs}.
\newblock PhD thesis, University of Illinois at Urbana-Champaign, 1996.

\bibitem{Ramalingan02}
G.~Ramalingam.
\newblock On sparse evaluation representations.
\newblock {\em Theoretical Computer Science}, 277(1-2):119--147, 2002.

\bibitem{Rimsa11}
Andrei~Alves Rimsa, Marcelo D'Amorim, and Fernando M.~Q. Pereira.
\newblock Tainted flow analysis on {e-SSA}-form programs.
\newblock In {\em CC}, pages 124--143. Springer, 2011.

\bibitem{Rimsa14}
Andrei~Alves Rimsa, Marcelo D'Amorim, Fernando M.~Q. Pereira, and Roberto
  Bigonha.
\newblock Efficient static checker for tainted variable attacks.
\newblock {\em Science of Computer Programming}, 80:91--105, 2014.

\bibitem{Rodrigues13}
Raphael~Ernani Rodrigues, Victor Hugo~Sperle Campos, and Fernando Magno~Quintao
  Pereira.
\newblock A fast and low overhead technique to secure programs against integer
  overflows.
\newblock In {\em CGO}, pages 1--11. ACM, 2013.

\bibitem{Roy10}
Subhajit Roy and Y.~N. Srikant.
\newblock The hot path ssa form: Extending the static single assignment form
  for speculative optimizations.
\newblock In {\em CC}, pages 304--323, 2010.

\bibitem{Scholz08}
Bernhard Scholz, Chenyi Zhang, and Cristina Cifuentes.
\newblock User-input dependence analysis via graph reachability.
\newblock Technical report, Sun, Inc., 2008.

\bibitem{Singer06}
Jeremy Singer.
\newblock {\em Static Program Analysis Based on Virtual Register Renaming}.
\newblock PhD thesis, University of Cambridge, 2006.

\bibitem{Sreedhar99}
Vugranam~C. Sreedhar, Roy~Dz ching Ju, David~M. Gillies, and Vatsa Santhanam.
\newblock Translating out of static single assignment form.
\newblock In {\em SAS}, pages 194--210. Springer-Verlag, 1999.

\bibitem{Stephenson00}
Mark Stephenson, Jonathan Babb, and Saman Amarasinghe.
\newblock Bitwidth analysis with application to silicon compilation.
\newblock In {\em PLDI}, pages 108--120. ACM, 2000.

\bibitem{Su05}
Zhendong Su and David Wagner.
\newblock A class of polynomially solvable range constraints for interval
  analysis without widenings.
\newblock {\em Theoretical Computeter Science}, 345(1):122--138, 2005.

\bibitem{Hochstadt08}
Sam Tobin-Hochstadt and Matthias Felleisen.
\newblock The design and implementation of typed scheme.
\newblock {\em POPL}, pages 395--406, 2008.

\bibitem{Wegman91}
Mark~N. Wegman and F.~Kenneth Zadeck.
\newblock Constant propagation with conditional branches.
\newblock {\em TOPLAS}, 13(2), 1991.

\bibitem{Weiss92}
Michael Weiss.
\newblock The transitive closure of control dependence: the iterated join.
\newblock {\em TOPLAS}, 1(2):178--190, 1992.

\bibitem{Zadeck84}
Frank~Kenneth Zadeck.
\newblock {\em Incremental Data Flow Analysis in a Structured Program Editor}.
\newblock PhD thesis, Rice University, 1984.

\end{thebibliography}
\end{small}

\newpage
\appendix

\section{Isomorphism to Sparse Evaluation Graphs}
\label{app:iso}

Given a control flow graph $G$, Choi {\em et al.} define a sparse evaluation graph as a tuple $\langle N_{SG}, E_{SG}, M \rangle$, such that:
\begin{itemize}
\item $N_{SG}$ is a set of nodes defined as follows:
\begin{enumerate}
\item $N_{SG}$ contains a node $n_s$ representing the entry \splitpoint $s \in G$; \item $N_{SG}$ contains a node $n_p$ for each \splitpoint $p \in G$ that is associated with a non-identity transfer function.
\item $N_{SG}$ contains a node $n_m$ for each point $m$ in the iterated dominance frontier of the \splitpoints of $G$ used to build the nodes in step (1) and (2).
These are called {\em meet} nodes.
\end{enumerate}
\item We let $P$ denote the set of \splitpoints $p \in G$ used in step 2 above, plus the \splitpoint $s \in G$ used in step 1 above;
we let $M$ denote the set of \splitpoints $m \in G$ used in step 3 above;
if we let $S=P\bigcup M$ then we define $E_{SG}$ as follows:
\begin{enumerate}
\item there is an edge $(n_q,n_m)\in N_{SG}^2$ whenever $m\in M$ and $q$ is, among all the nodes in $S$, the immediate dominator of one of the CFG predecessors of $m$.
See \texttt{search(3b)} and \texttt{link(2b)} in Choi {\em et al}~\cite{Choi91}; 
\item there is an edge $(n_q,n_p)\in N_{SG}^2$ whenever $p \in P$, and $q$ is, among all the nodes in $S$, the immediate dominator of $p$.
See \texttt{search(1)} and \texttt{link(2b)}~\cite{Choi91}; 
\end{enumerate} 
\item The mapping function $M:
E_{G} \mapsto N_{SG}$ associates to each edge $(u,v)$ of the CFG the node $n_q \in N_{SG}$, whenever $q \in S$ is the immediate dominator of $u \in G$.
See \texttt{search(3a)}~\cite{Choi91}.
This is done through the recursive function \texttt{search} that performs a topological traversal of the CFG (DFS of the dominance tree; See \texttt{search(4)}~\cite{Choi91}).
\end{itemize}
Theorem~\ref{theo:iso} states that, for forward partitioned variable data-flow problems (PVP), the algorithm in Figure~\ref{fig:SSIfy} can build program representations isomorphic to Sparse Evaluation Graphs.
The proof that this result holds for backward data-flow problems, is analogous, and we omit it.

\begin{lemma}[CFG cover]
\label{lem:part} Let $Prog$ be a program with its corresponding CFG $G$ with start node $s$, and exit node $x$.
Let $Prog'$ be the program that we obtain from $Prog$ by:
\begin{enumerate}
\item adding a pseudo-definition of each variable to $s$; 
\item adding a pseudo-use of each variable to $x$; 
\item placing a pseudo-use of a variable $v$ at each \splitpoint where $v$ is defined; 
\item converting the resulting program into SSA form.
\end{enumerate}
If $v$ is a variable in $Prog$, then the live ranges of the different names of $v$ in $Prog'$ completely partition the \progpoints of $G$.
In other words, each \progpoint of $G$ belongs to exactly one live range of $v$ in $Prog'$.
\end{lemma}

\begin{proof}
First, $v$ is alive at every \progpoint of $G$, due to transformations (1), (2) and (3).
Therefore, if $V$ is the set of the different names of $v$ after the conversion to SSA form in step (4), then any \progpoint of $G$ belongs to the live range of at least one $v' \in V$.
The result follows from a well-know property of Cytron's SSA-form conversion algorithm~\cite{Cytron91}, which, as observed by Sreedhar {\em et al.}~\cite{Sreedhar99}, creates variables with non-intersecting live ranges.
In other words, after the SSA renaming, two different names of $v$ cannot be simultaneously alive at a \progpoint $p$.
\end{proof}

~\\
\noindent
\textbf{[Equivalence SSI/SEG - See Theorem 3]}
Given a forward Sparse Evaluation Graph ($SEG$) that represents a variable $v$ in a program representation $Prog$ with CFG $G$, there exits a live range splitting strategy that once applied on $v$ builds a program representation that is isomorphic to $SEG$.

\begin{proof}
We argue that the SEG of $v$ is isomorphic to the representation of $v$ in $Prog'$, the program representation that we derive from $Prog$ by applying the transformations 1-3 listed in Lemma~\ref{lem:part} in addition to a pass of \SSIfy.
If we let $P$, as before, be defined as the set of CFG nodes associated with non-identity transfer functions, plus the start node $s$ of the CFG, then after we apply the splitting strategy $P_\downarrow$, we have that:
\begin{enumerate}
\item there will be exactly one definition per node of $P$ and one  definition per node of $DF^{+}(P)$.
So there is an one-to-one correspondence between SSA definitions and SG nodes.
\item From Lemma~\ref{lem:part} the live-ranges of the different names of $v$ provides a partitioning of the \progpoints of $G$.
If $v'$ is a new name of $v$, then each \progpoint where $v'$ is alive is dominated by $v'$'s definition\footnote{This is a classical result of SSA-form.
See Budimlic {\em et al.}~\cite{Budimlic02} for a proof}.
Each \progpoint belongs to the live-range of the name of $v$ whose definition immediately dominates it (among all definitions).
Thus, live ranges give origin to a function that maps SSA definitions to \progpoints.
Consequently, there is an isomorphism between the live-ranges and the mapping function $M$.
\item def-use chains on $Prog'$ are isomorphic to the edges in $E_{SG}$:
indeed a SEG node $n_p$ is linked to $n_q$ whenever (i) $n_p$ immediately dominates $n_q$ if $q\in P$; or (ii) $n_q$ is in the dominance frontier of $n_p$ if $q\in M$.
In the former case the definition of $v$ at $p$ reaches the (pseudo-)use of $v$ at $q$.
In the latter this definition reaches the use of $v$ at the $\phi$-function placed at $q$ by $\SSIfy(v, P_\downarrow)$.
\end{enumerate}
\end{proof}

In the proof of Theorem~\ref{theo:iso} we had to augment the program with a pseudo-definition of $v$ at the CFG's entry node and a pseudo-use at every actual definition of $v$ and at the CFG's exit node.
The difference between a code with or without pseudo uses/defs is related to the necessity to compute data-flow information beyond the live-ranges of variables or not.
This necessity exists for optimizations such as partial redundancy elimination, which may move, create or delete code.

Figure~\ref{fig:ssiChoiReaching} compares SEG and the forward live range splitting strategy in the example taken from Figure 11 of Choi {\em et al.}~\cite{Choi91}, which shows the reaching uses analysis.
In the left we see the original program, and in the middle the SEG built for a forward flow analysis that extracts information from uses of variables.
We have augmented the edges in the left CFG with the mapping $M$ of SEG nodes to CFG edges.
In the right we see the same CFG, augmented with pseudo defs and uses, after been transformed by $\SSIfy$ applied on the \splitpoints $\{S, 4, 5, 7, 11, 12\}_\downarrow$.
The edges of this CFG are labeled with the definitions of $v$ live there.

\begin{figure}[t!]
\begin{center}
\includegraphics[width=\textwidth]{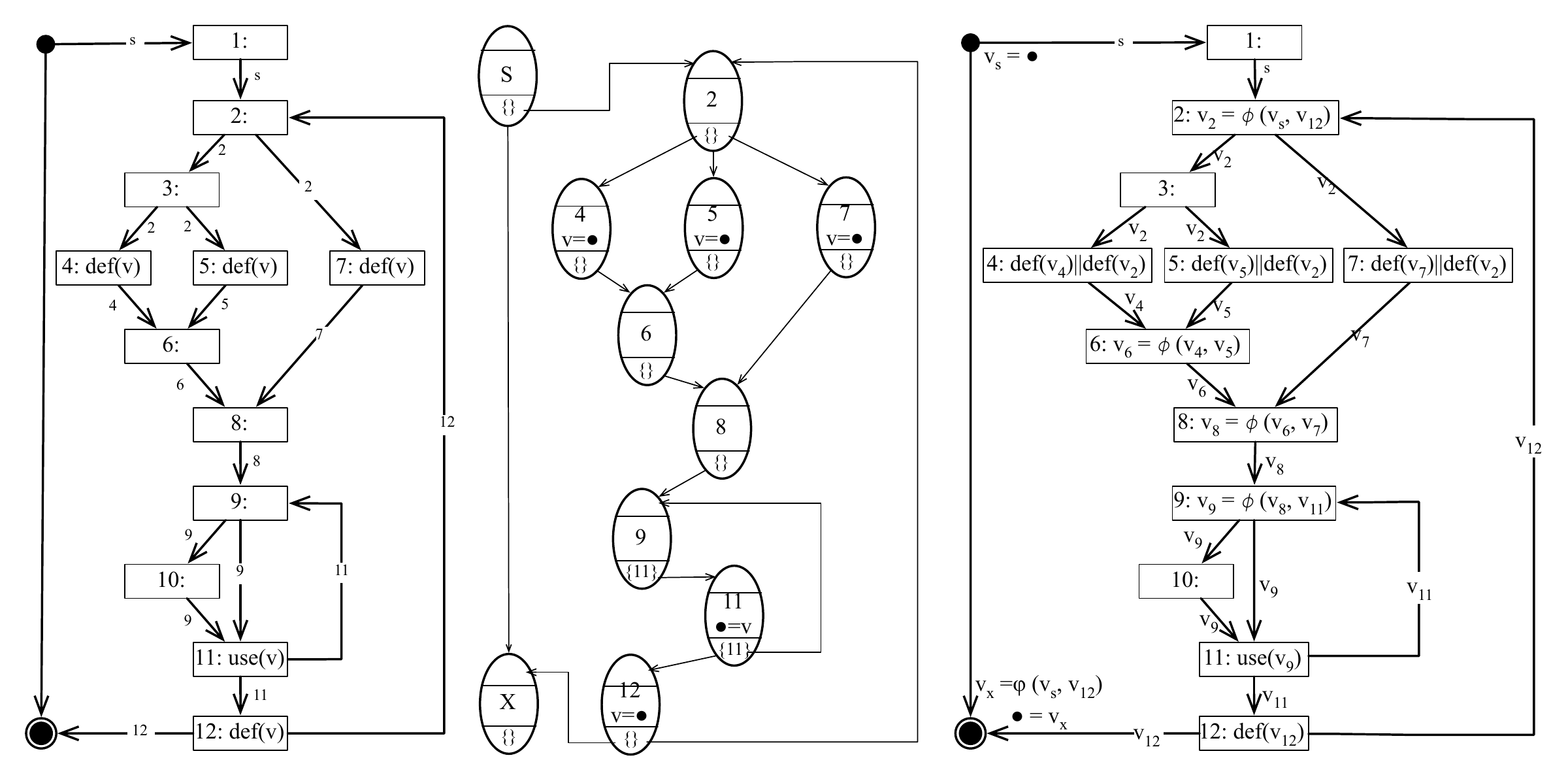} \caption{Example of equivalence between SEGs and our live range splitting strategy for reaching uses.} \label{fig:ssiChoiReaching} \end{center} \end{figure}

\section{Correctness of our SSIfication}
\label{sec:ssify}

In this section we consider a unidirectional forward (resp. backward) PLV problem stated as a set of equations $[v]^p = [v]^p \wedge F_v^{s,p}(\dots)$ for every variable $v$, each \progpoint $p$, and each $s \in \mathit{preds}(p)$ (resp. $s \in \mathit{succs}(p)$). 
We rely on the nomenclature introduced by Definition~\ref{def:tf} in order
to prove Theorem~\ref{theo:ssify}.

\begin{lemma}[Live range preservation]
\label{lem:live_pres}
If variable $v$ is live at a \progpoint $p$, then there is a version of $v$ live at $p$ after we run {\sf SSIfy}.
\end{lemma}

\begin{proof}
{\sf Split} cannot remove any live range of $v$, as it only inserts ``copies" from $v$ to $v$, e.g., each copy has the same source and destination.
{\sf Rename} removes live ranges of $v$, but it replaces them with the live ranges of new versions of this variable whenever a use of $v$ is renamed.
{\sf Clean} only removes ``copies"; hence, all the original instructions remain in the code.
\end{proof}

\begin{lemma}[Non-Overlapping]
\label{lem:live_over}
Two different versions of $v$, e.g., $v_k$ and  $v_j$ cannot both be live at a \progpoint $p$ transformed by {\sf SSIfy}.
\end{lemma}

\begin{proof}
The only algorithm that creates new versions of $v$ is {\sf rename}.
Each new version of $v$ is unique, as we ensure in lines 28-30 of the algorithm.
If {\sf rename} changes the use of $v$ to $v_k$ at a \splitpoint $i$, then there exists a definition of $v_k$ at some \splitpoint $i'$ that dominates $i$, as we ensure in line 22 of the algorithm.
Let us assume that we have two versions of $v$, e.g., $v_k$ and $v_j$, live at a \progpoint $p$, in order to derive a contradiction.
In this case, there exist \splitpoints $i_k$ where $v_k$ is used, and $i_j$ where $v_j$ is used, reachable from $p$.
Also there exists a \splitpoint $i_k'$ where $v_k$ is defined, and a \splitpoint $i_j'$ where $v_j$ is defined. $i_k'$ dominates $p$, and $i_j'$ dominates $p$.
Thus, either $i_k'$ dominates $i_j'$ or vice-versa.
Without loss of generality, let us assume that $i_k'$ dominates $i_j'$.
In this case, {\sf rename} visits $i_k'$ first, and upon visiting $i_j'$, places the definition of $v_j$ on top of the definition of $v_k$ in the stack in line 31.
Thus, $i_k'$ cannot dominate $i'_j$, or we would have, at $i_k$, a use of $v_j$, instead of $v_k$.
\end{proof}

~\\
\noindent
\textbf{[Semantics - Theorem 1]}
{\sf SSIfy} maintains the following property:
if a value $n$ written to variable $v$ at \splitpoint $i'$ is read at a \splitpoint $i$ in the original program, then the same value assigned to a version of variable $v$ at \splitpoint $i'$ is read at a \splitpoint $i$ after transformation.

\begin{proof}
For simplicity, we will extend the meaning of ``copy'' to include not only
the parallel copies placed at interior nodes, but also $\phi$ and $\sigma$-functions.
{\sf Split} cannot create new values, as it only inserts ``copies".
{\sf Clean} cannot remove values, as it only removes ``copies".
From the hypothesis we know that the definition of $v$ that reaches $i$ is
live at $i$.
From Lemma~\ref{lem:live_pres} we know that there is a version of v live at $i$.
From Lemma~\ref{lem:live_over} we know that only one version of $v$ can be live at $i$, and so {\sf rename} cannot send new values to $i$.
\end{proof}

Now suppose that the program, not necessarily under SSI form, fulfills {INFO} and {LINK} from Definition~\ref{def:ssi} for a system of monotone equations $E_{\var{dense}}$, given as a set of constraints $[v]^p \inclu F_v^{s,p}([v_1]^s, \dots, [v_n]^s)$. 
Consider a live range splitting strategy $\SS_v$ that \emph{includes} for each variable $v$ the set of \splitpoints $I_\downarrow$ (resp. $I_\uparrow$) where $F_v^{s,p}$ is non-trivial.
The following theorem states that Algorithm~\textsf{\SSIfy} creates a program form that fulfills the Static Single Information property.

~\\
\noindent
\textbf{[Correctness of SSIfy - Theorem 2]}
Given the conditions stated above, Algorithm~\textsf{\SSIfy}($v,\,\SS_v$) creates a new program representation such that:
\begin{enumerate}
\item there exists a system of equations $E^{\var{ssi}}_{\var{dense}}$, isomorphic to $E_{\var{dense}}$ for which the new program representation fulfills the SSI property.
\item if $E_{\var{dense}}$ is monotone then $E^{\var{ssi}}_{\var{dense}}$ is also monotone.
\end{enumerate}

\begin{proof}
We derive from this new program representation a system of equations isomorphic to the initial one by associating trivial transfer functions with the newly created ``copies''.
The {INFO} and {LINK} properties are trivially maintained.
As only trivial and constant functions have been added, monotonicity is maintained.

To show that we provide {SPLIT-DEF}, we must first show that each $i\in \textrm{live}(v)$ where $F_v^s$ is non-trivial contains a definition (resp. last use) of $v$.
The function {\sf split} separates these \progpoints in lines 9 and 16, and later, in line 23, inserts definitions in those \splitpoints.
To show that we provide {SPLIT-MEET}, we must prove that each join (resp. split) node for which $E_{\var{dense}}$ has possibly different values on its incoming edges should have a $\phi$-function (resp. $\sigma$-function) for $v$.
These \progpoints are separated in lines 7 and 14 of {\sf split}.
To see why this is the case, notice that line 7 separates the \progpoints in the iterated dominance frontier of \progpoints that originate information that flows forward.
These are, as a direct consequence of the definition of iterated dominance frontier, the \splitpoints where information collide.
Similarly, line 14 separates the \progpoints in the post-dominance frontier of regions which originate information that flows backwardly.

We ensure {VERSION} as a consequence of the SSA conversion.
All our program representations preserve the SSA representation, as we include the definition sites of $v$ in line 11 of {\sf split}.
Function {\sf rename} ensures the existence of only one definition of each variable in the program code (line 27), and that each definition dominates all its uses (consequence of the traversal order).
Therefore, the newly created live ranges are connected on the dominance tree of the source program.
Function {\sf rename} also creates a new program representation for which it is straightforward to build a system of equations $E^{\var{ssi}}_{\var{dense}}$ isomorphic to $E_{\var{dense}}$:
Firstly, the constraint variables are renamed in the same way that program variables are.
Secondly, for each program variable, new system variables bound to $\bot$ are created for each \progpoint outside of its live-range.

\end{proof}

\section{Equivalence between sparse and dense analyses.}
\label{app:sparse_equiv}
We have shown that {\sf SSIfy} transforms a program $P$ into another program $P^{\var{ssi}}$ with the same semantics.
Furthermore, this representation provides the SSI property for a system of equations $E^{\var{ssi}}_{\var{dense}}$ that we extract from $P^{\var{ssi}}$.
This system is isomorphic to the system of equations $E_{\var{dense}}$ that we extract from $P$.
From the so obtained program under SSI for the constrained system $E^{\var{ssi}}_{\var{dense}}$, Definition~\ref{def:ssi_eq} shows how to construct a sparse constrained system $E^{\var{ssi}}_{\var{sparse}}$.
When transfer functions are monotone and the lattice has finite height, Theorem~\ref{thm:sparse_equiv} states the equivalence between the sparse and the dense systems. The purpose of this section is to prove this theorem. We start by introducing the notion of {\em coalescing}.
Let $E$ be a constraint system that associates with each $1\leq i\leq n$ the constraint $a_i\inclu H_i(a_1,\dots, a_n)$, where each $a_i$ is an element of a lattice $\cal L$ of finite height, and $H_i$ is a monotone function from ${\cal L}^n$ to $\cal L$.
Let $(A_1,\dots,A_n)$ be the maximum solution to this system, and let $1\leq m\leq n$ such that $\forall i,\, 1\leq i\leq m,\, A_i=A_m$.
We define a ``coalesced" constraint system $E_{coal}$ in the following way:
for each $1\leq i\leq m$ we create the constraint $b_m\inclu H_i(b_m,\dots, b_m, b_{m+1},\dots, b_n)$; for each $m< i\leq n$ we create the constraint $b_i\inclu H_i(b_m,\dots, b_m, b_{m+1},\dots, b_n)$.
Lemma~\ref{lemma:coal_equiv} shows that coalescing preserves the maximum
solution of the original system.

\begin{lemma}[Equivalence with coalescing]
\label{lemma:coal_equiv}
If $E$ is a constraint system with maximum solution $(A_1, \dots, A_m, \ldots, A_n)$, for any $i, j, 1 \leq i, j \leq m$ we have that $A_i = A_j$, and $E_{coal}$ is the ``coalesced"~system that we derive from $E$, then the maximum solution of $E_{coal}$ is $(A_m,\dots,A_n)$.
\end{lemma}

\begin{proof}
Both system have a (unique) maximum solution~(see e.g.~\cite{Nielson05}), although the solution of the ``coalesced" system has smaller cardinality, e.g., n-m+1.
Now, as $(A_m, \dots, A_m, A_{m+1}, \dots, A_n)$ is a solution to $E$, by definition of $E_{coal}$, $(A_m, \dots, A_n)$ is a solution to $E_{coal}$. Let us prove that this solution is maximum, i.e. for any solution $(B_m, \dots, B_n)$ of $E_{coal}$, we have $(B_m, \dots, B_n) \sqsubseteq (A_m, \dots, A_n)$. By definition of $E_{coal}$, we have that $(B_m, \dots, B_m, B_{m+1}, \dots, B_n)$ is a solution to $E$. As $(A_1, \dots, A_n)$ is maximum, we have
$(B_m,  \dots, B_m, B_{m+1}, \dots, $ $B_n)\sqsubseteq (A_1,\dots,A_n)$. So $(B_m,\dots,B_n)\sqsubseteq (A_m,\dots,A_n)$. 
\end{proof}

We now prove Theorem~\ref{thm:sparse_equiv}, which states that there exists a direct mapping between the maximum solution of a dense constraint system associated with a SSI-form program, and the sparse system that we can derive from it, according to Definition~\ref{def:ssi_eq}.

\begin{theorem}[sparse $\equiv$ dense]
\label{thm:sparse_equiv}
Consider a program in SSI-form that gives origin to a constraint system $E^{\var{ssi}}_{\var{dense}}$ associating with each variable $v$ the constraints $[v]^p = [v]^p \wedge F_v^{s,p}([v_1]^s, \dots, [v_n]^s)$. Suppose that each $F_v^{s,p}$ is a monotone function from ${\cal L}^n$ to $\cal L$ where $\cal L$ is of finite height.
Let $(Y_v)_{v\in variables}$ be the maximum solution of the corresponding sparse constraint system.

Then, $(X_v^p)_{(v,i)\in variables\times prog\_points}$ with
$\strut\left\lbrace
\begin{tabular}{ll}
$X_v^p=Y_v$ & for $p\in \var{live}(v)$\\
 $X_v^p=\bot$ & otherwise
\end{tabular}
\right.
$
is the maximum solution to $E_{\var{dense}}^{\var{ssi}}$.
\end{theorem}

\begin{proof}
The constraint systems $E_{\var{dense}}^{\var{ssi}}$ and $E_{\var{sparse}}^{\var{ssi}}$ have a maximum unique solution, because the transfer functions are monotone and  $\cal{L}$ has finite height

The idea of the proof is to modify the constraint system $E_{\var{dense}}^{\var{ssi}}$ into a system equivalent to $E_{\var{sparse}}^{\var{ssi}}$.
To accomplish this transformation, we
(i) replace each $F_v^{s,p}$ by $G_v^i$, where $G_v^i$ is constructed as in Definition~\ref{def:ssi_eq};
(ii) for each $v$, coalesce $\lbrace[v]^p\rbrace_{p\in \var{live}(v)}$ into $[v]$;
(iii) coalesce all other constraint variables into $[\undef]$.

The {LINK} property allows us to replace $F_v^{s,p}$ by $G_v^i$. Due to {SPLIT-DEF},
a new variable is defined at each \splitpoint where information is generated, and due to {VERSION} there is only one live range associated with each variable. Hence, $\lbrace[v]^p\rbrace_{p\in \var{live}(v)}$ is invariant. Due to {INFO}, we have that  $\lbrace[v]^p\rbrace_{p\not\in \var{live}(v)}$ is bound to $\bot$. Due to Lemma~\ref{lemma:coal_equiv}, we know that this new constraint system has a maximum solution $(Y_v)_{v\in \textit{variables}\cup \undef}$: $X_v^p$ equals $Y_v$ for all $p\in \var{live}(v)$, and $Y_\undef$ otherwise.

We translate each constraint $[v]^p\inclu F_v^{s,p}([v_1]^s,\dots,[v_n]^s)$ (with $i$ the \splitpoint between $p$ and $s$), in the original system, to a constraint in the ``coalesced'' one in the following way:
$$\left\lbrace
\begin{array}{lllr}
\textrm{if } p\in \var{live}(v): & \textrm{if } i\in \var{defs}(v) & :\, [v] \inclu G_v^i([a], \dots, [b]) & (1)\\
                           & \textrm{else} &  :\, [v] \inclu [v] & (2)\\
\textrm{otherwise}         &               &  :\, [\undef] \inclu \bot & (3)\\
\end{array}
\right.
$$
Case (1) follows from {LINK}, case (2) follows from {SPLIT-DEF}, and case (3) follows from {INFO}. By ignoring $\undef$ that appears only in (3), and by removing the constraints produced by (2), which are useless, we obtain $E_{\var{sparse}}^{\var{ssi}}$. 
\end{proof}

\end{document}